\documentclass[11pt]{article}

\usepackage[letterpaper,margin=0.75in]{geometry}
\usepackage[utf8]{inputenc}
\usepackage[T1]{fontenc}
\usepackage{mathptmx}
\usepackage{url}
\usepackage{microtype}
\usepackage{amsmath}
\usepackage{amssymb}
\usepackage{bm}
\usepackage{amsthm}
\usepackage{algorithm}
\usepackage{algorithmic}
\usepackage{float}
\usepackage{array}
\usepackage{booktabs}
\usepackage{graphicx}
\usepackage{subcaption}
\usepackage{xcolor}
\usepackage[hidelinks,hypertexnames=false,bookmarks=false]{hyperref}
\graphicspath{{figures/}{./figures/}}

\newcommand{\FloatBarrier}{}

\newtheorem{theorem}{Theorem}
\newtheorem{proposition}{Proposition}

\newtheorem{definition}{Definition}

\newtheorem{corollary}{Corollary}

\DeclareMathOperator{\Var}{Var}
\DeclareMathOperator{\Cov}{Cov}
\DeclareMathOperator{\tr}{tr}

\newcolumntype{L}[1]{>{\raggedright\arraybackslash}p{#1}}

\newcommand{\localcaauthorversion}{}

\title{\Large\bfseries Shunting Inhibition and Dendritic Branching\\
Shape Local Credit Assignment}
\author{%
  Houman Safaai\,$^{1}$\thanks{Correspondence: \texttt{houman\_safaai@harvard.edu}; \texttt{bernardo\_sabatini@hms.harvard.edu}.} \quad Maceo Richards\,$^{1}$ \quad Bernardo L. Sabatini\,$^{1,2}$\textsuperscript{*} \\[0.5em]
  $^{1}$ Kempner Institute for the Study of Natural and Artificial Intelligence \\
  at Harvard University \\
  $^{2}$ Department of Neurobiology, Howard Hughes Medical Institute, \\
  Harvard Medical School, Boston, MA 02115, USA
}
\date{}

\begin{document}
\maketitle

\begin{abstract}
    Biological neurons assign credit across branching dendrites, where synaptic drive, dendritic conductance, local voltage, and somatic teaching signals interact to shape synaptic plasticity. We study conductance-based dendritic networks with excitatory and inhibitory synapses, shunting inhibition, and tree-structured branch-to-soma coupling. We examine the conditions under which restricted somatic feedback can approximate compartment-specific backpropagated errors. Exact gradients factor into \emph{local eligibility $\times$ compartment error} terms: the eligibility is set by presynaptic activity, driving force, and input resistance, whereas the fast non-local term is a path-specific error obtained by transporting a soma error through dendritic gains. This factorization turns local learning into a credit-signal approximation problem. We test the hypothesis that shunting inhibition benefits learning under these constraints when it reshapes the compartment-error field to better match restricted somatic feedback. Exact-gradient reconstruction verifies the factorization; path-gain, feedback-fidelity, inhibition-intervention, and transported-error-oracle diagnostics test the proposed mechanism and its limits. Under nonnegative conductances and a 5-factor (5F) rule with matched-width/scalar-fallback feedback, shunting local credit assignment (LocalCA) remains $5$--$6$ percentage points below matched backpropagation on MNIST, Fashion-MNIST, and figure-ground MNIST, indicating that feedback-field fidelity remains a major bottleneck. Additional controls show that a 3-factor rule approaches matched backpropagation with exact transported feedback in the shunting model and with neuron-wise feedback in both architectures. However, shunting has no general advantage under matched initialization. These results show how conductance and dendritic branching enter the exact credit equation, and identify restricted feedback as a principal limit in these experiments.
    \end{abstract}

    \section{Introduction}
    \label{sec:intro}
    
    Unlike point units in standard neural networks, biological neurons assign credit to synapses distributed across branching dendritic trees. Each synapse receives presynaptic drive and samples local membrane voltage, synaptic reversal potential, and total conductance (the reciprocal of the local input resistance). Inhibitory synapses increase total local conductance, counteracting local excitation by shunting voltage. This increase in branch conductance alters the gain of every dendritic path passing through that compartment.
    
    We examine whether these biophysical ingredients can facilitate local credit assignment (LocalCA). Exact backpropagation \cite{rumelhart1986learning} assigns a distinct error to every dendritic compartment, whereas a biologically plausible supervisory signal is more likely to arrive as a low-bandwidth somatic or modulatory broadcast \cite{richards2019dendritic}. The central question is therefore not whether such a broadcast can reproduce unconstrained backpropagation in every setting, but when the exact dendritic error field is simple enough to be approximated by a low-bandwidth broadcast.
    Our contribution is not simply to add dendritic structure to a neural network, but to show how conductance-based synapses, shunting inhibition, and dendritic topology shape the geometry of the credit signal.
    
    \begin{figure}[t]
    \centering
    \includegraphics[width=\textwidth]{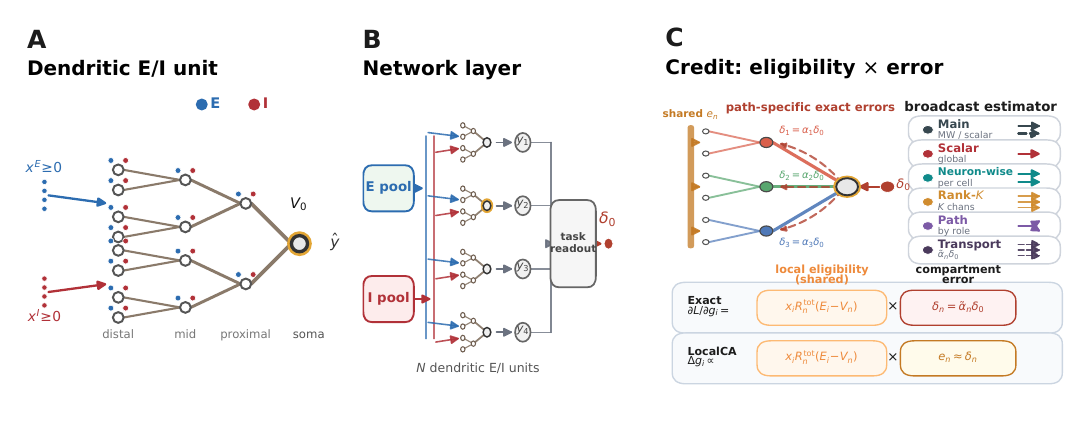}
    \caption{\textbf{Model and credit assignment.}
    (A)~Schematic of a dendritic unit with excitatory and inhibitory inputs on each branch. (B)~Layer of such units projecting to a task readout; $\boldsymbol{\delta}_0=\partial L/\partial\mathbf{V}_0$ is the soma/core error used for feedback. (C)~Exact updates multiply local eligibility by path-specific errors $\delta_{u,n}$; LocalCA replaces them with broadcast estimates $e_{u,n}$. The main matched-width/scalar-fallback hybrid is shown separately from strict scalar, neuron-wise, low-rank, path-structured, and transported feedback.}
    \label{fig:model_schematic}
    \end{figure}
    \FloatBarrier
    
    Starting from conductance-based dendritic voltage equations \cite{koch1999biophysics}, we derive exact gradients for dendritic trees (Fig.~\ref{fig:model_schematic}; Theorem~\ref{thm:tree_backprop}). Each synaptic gradient factorizes into a synapse-local eligibility term and a single non-local compartment error:
    \[
    \text{gradient}=\underbrace{\text{local eligibility}}_{\text{presynaptic drive, }E{-}V,\;R^{\mathrm{tot}}}
    \times
    \underbrace{\text{compartment error}}_{\text{path-specific non-local term}}.
    \]
    For compartment $n$ of neuron $u$, the exact compartment error equals that neuron's somatic error multiplied by the voltage gain along the unique path from $n$ to its soma. This path gain is the product of the local child-to-parent transfer factors on that route. LocalCA preserves the exact local eligibility and replaces this path-specific error with a broadcast field.
    
    \noindent\textbf{What is local in LocalCA?}
    The local eligibility contains three factors: presynaptic drive, synaptic driving force (the difference between synaptic reversal potential and local voltage), and input resistance. The task-dependent compartment error is approximated by a broadcast $e_n$. The 3F rule multiplies these terms, whereas 4F and 5F add slowly estimated branch-level preconditioners without changing the fast eligibility used by 3F.
    
    This factorization turns LocalCA into a feedback-compatibility problem: low intrinsic rank is insufficient unless the exact error field aligns with the available feedback. We compare a global scalar, the main matched-width/scalar-fallback hybrid, ancestry-shared, random low-rank, path-structured, and exact transported feedback. Shunting changes input resistance along dendritic paths, but helps learning only when that change improves alignment with the restricted broadcast.
    
    \noindent\textbf{Related work.}
    The work builds on models of dendritic integration, nonlinear branch computation, normalization, three-factor plasticity, and compartmental teaching signals \cite{koch1983nonlinear,koch1999biophysics,poirazi2003pyramidal,london2005dendritic,silver2010neuronal,carandini2012normalization,fremaux2016threefactor,urbanczik2014dendritic,guerguiev2017segregated,sacramento2018dendritic,payeur2021burst,iyer2022activedendrites}. It also relates to random, predictive, equilibrium, perturbation, and broadcast alternatives to backpropagation and to reviews of how the brain might approximate it \cite{lillicrap2016random,nokland2016dfa,lee2015dtp,scellier2017equilibrium,meulemans2021dfc,millidge2021predictive,hinton2022forward,dellaferrera2022pepita,haider2021latent,max2024backprojections,song2024prospective,whittington2019theories,lillicrap2020backprop}. Recent dendritic recordings report neuron-specific instructive signals during learning \cite{francioni2026vectorized}, providing a biological example related to the neuron-indexed feedback studied here. Rather than posit a compartmental teaching signal, we derive the exact credit required by a conductance tree and test its approximation from soma-level feedback. Appendix Table~\ref{tab:related_local_learning} gives a detailed comparison across approaches, and Appendix Fig.~\ref{fig:fa_dfa_appendix} shows the implemented FA/DFA reference.
    
    Exact reconstruction verifies the factorization, and dendritic-block metrics avoid masking it with an aligned soma coupling (Fig.~\ref{fig:gradient_fidelity}). An initial three-checkpoint cohort suggested better branch-gradient direction under shunting, but an independent five-seed replication did not reproduce that ordering. Reactivation and initialization controls show that the ordering is sensitive to the models' forward activation states and is not an architecture-wide advantage. In contrast, exact transport and neuron-wise feedback consistently reduce the learning gap, identifying feedback construction as the principal bottleneck to good performance.
    
    \noindent\textbf{Contributions.}
    We derive the conductance-stage path gain, its local-eligibility/path-error factorization, and an inhibitory path-gain ratio; define path-gain and feedback-compatibility diagnostics; and identify the restricted-feedback conditions under which shunting helps or fails.
    
    \noindent\textbf{Scope of evidence.}
    Theorems concern steady-state conductance trees. Simulations test reconstruction, feedback fidelity, oracle transport, and matched additive/shunting models; transported errors are diagnostic upper bounds, not proposed biological signals.
    
    \begin{figure}[t]
    \centering
    \includegraphics[width=\textwidth]{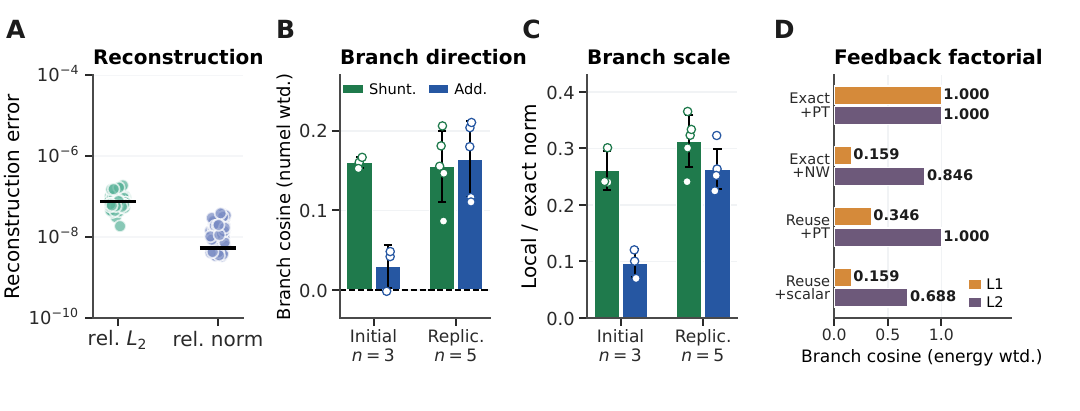}
    \caption{\textbf{Exact reconstruction and gradient diagnostics.}
    (A) Exact-factorization reconstruction against autograd: relative $L_2$ error and relative norm error, one point per parameter tensor. (B,C) Matched 3F diagnostics at non-somatic dendritic parameter blocks for the initial three-checkpoint cohort and an independent five-seed replication. Direction is the parameter-count-weighted cosine between local and exact branch gradients ($1$ means parallel; $0$ means orthogonal). Scale is the local-to-exact gradient-norm ratio ($1$ means equal norm). Each dot is one checkpoint averaged over train, validation, and test diagnostic batches; bars show mean $\pm$1 s.d.\ across seeds. Full values and inference limits are in Appendix Tables~\ref{tab:matched_3f_gradient} and \ref{tab:fresh_3f_replication}. (D) Energy-weighted branch-gradient cosine for selected combinations of layer-soma teaching signal and within-tree feedback. Exact/Reuse specifies whether the layer receives its exact soma error or reuses the final-core teaching vector; PT, NW, and scalar specify exact path transport, neuron-wise sharing over descendants, and scalar fallback, respectively. L1 and L2 are the two dendritic network layers. The two L1 conditions labeled $0.159$ differ before rounding. Metrics and protocol: Appendix~\ref{app:extra_results}.}
    \label{fig:gradient_fidelity}
    \end{figure}
    
    \section{Compartmental Voltage Model and Gradient Derivation}
    \label{sec:model}
    
    For each dendritic compartment, we use a steady-state conductance model derived from discretized passive cable dynamics \cite{koch1999biophysics,dayan2001theoretical}.
    \noindent\emph{Voltage equation and local sensitivities.}
    Consider compartment $n$ with synaptic inputs $j$ (activity $x_j$, reversal $E_j$, conductance $g_j^{\mathrm{syn}}\!\geq\!0$) and dendritic inputs from children (voltage $V_j$, transmitted activity $a_j=f_j(V_j)$, conductance $g_j^{\mathrm{den}}\!\geq\!0$). We use this reactivation notation throughout; when reactivation is disabled, $f_j$ is the identity, so $a_j=V_j$ and $f_j'(V_j)=1$.
    In normalized units, leak and inhibitory reversal potentials are $0$, the excitatory reversal potential is $1$, and leak conductance is fixed to $1$ (Appendix Table~\ref{tab:units_parameterization}). The steady-state voltage is:
    \begin{equation}
    V_n = \frac{\sum_j E_j x_j g_j^{\mathrm{syn}} + \sum_j a_j g_j^{\mathrm{den}}}{\underbrace{\sum_j x_j g_j^{\mathrm{syn}} + \sum_j g_j^{\mathrm{den}} + 1}_{g_n^{\mathrm{tot}}}},
    \qquad R_n^{\mathrm{tot}} = 1/g_n^{\mathrm{tot}}.
    \label{eq:voltage}
    \end{equation}
    $V_n$ is a convex combination of reversal potentials, transmitted child activities, and leak. Writing $\mathcal{S}_n = \{0\} \cup \{E_j\}_j \cup \{a_j\}_j$, we have $\min\mathcal{S}_n \leq V_n \leq \max\mathcal{S}_n$ and $0 < R_n^{\mathrm{tot}} \leq 1$.
    The local sensitivities follow directly:
    
    \begin{proposition}[Local Sensitivities]
    \label{prop:sensitivities}
    \begin{equation}
    \frac{\partial V_n}{\partial g_i^{\mathrm{syn}}} = x_i R_n^{\mathrm{tot}} (E_i - V_n),
    \quad
    \frac{\partial V_n}{\partial V_i} = f_i'(V_i)g_i^{\mathrm{den}} R_n^{\mathrm{tot}},
    \quad
    \frac{\partial V_n}{\partial g_i^{\mathrm{den}}} = R_n^{\mathrm{tot}} (a_i - V_n).
    \label{eq:sensitivities}
    \end{equation}
    \end{proposition}
    Eq.~\eqref{eq:sensitivities} gives the eligibility factors used by the local rules below: presynaptic activity or transmitted-child-activity difference, synaptic driving force, input resistance, and the local reactivation derivative. Setting $f_i'(V_i)=1$ recovers identity transfer.

    Each compartment therefore computes a conductance-weighted voltage: excitation pulls toward the excitatory reversal potential, leak and shunting inhibition pull toward zero, and dendritic coupling transfers child voltages toward the soma. Throughout the manuscript we write dendritic morphologies as rooted trees $[b_1,b_2,\dots,b_D]$, where each factor gives the fan-in from one dendritic stage to the next. Thus, $[3,3]$ has three proximal branches per soma and three distal branches per proximal branch, for nine distal leaves. Synaptic inputs terminate on branches rather than directly on the soma. Branches receive excitatory conductance-based inputs and, when enabled, learned input-driven inhibitory conductances; this is not a recurrent lateral inhibitory circuit. The additive and shunting cores are tree-matched and differ in branch-level integration. Branching creates path structure, conductances set local state, and shunting changes path gain through input resistance.

    \noindent\textbf{Notation.} A compact reference is provided in Appendix Table~\ref{tab:notation_reference}. We use $\delta_{0,u}^{(\ell),V}=\partial L/\partial V_{0,u}^{(\ell)}$ for the exact voltage-space soma error, $\delta_{0,u}^{(\ell),a}=\partial L/\partial a_{0,u}^{(\ell)}$ for its activation-space counterpart, $\boldsymbol{\delta}_0$ for the operational layer-wide teaching vector, $\bar{\delta}$ for its global scalar average, and $\delta_{u,n}^{(\ell)}=\partial L/\partial V_{u,n}^{(\ell)}$ for the compartment error. We use $\alpha_n^{\mathrm{cond}}$ for the conductance-only path gain and $\tilde{\alpha}_n$ for the effective path gain including reactivation derivatives; they coincide when reactivation is disabled. $N_E$ and $N_I$ denote the numbers of excitatory and inhibitory synapses per branch, $r_n^{\mathrm{4F}}$ the empirical 4F branch-soma covariance proxy, and $\phi_n$ the 5F bounded preconditioner. All conductances and presynaptic activities in the main setting are nonnegative; additive voltages may be signed comparators.
    
    \noindent\emph{Shunting inhibition as divisive gain control.}
    An inhibitory synapse with $E_{\mathrm{inh}}\!\approx\!0$ contributes current $(0{-}V_n)x_j g_j^{\mathrm{syn}}$ and increases $g_n^{\mathrm{tot}}$.
    Its sensitivity is $\partial V_n / \partial g_j^{\mathrm{syn}} = -x_j R_n^{\mathrm{tot}} V_n$, which corresponds to multiplicative attenuation (divisive normalization).
    Shunting is divisive at the voltage level, but its effect on firing rates can be subtractive in some regimes \cite{holt1997shunting}.
    Inhibitory plasticity can balance excitation dynamically \cite{vogels2011inhibitory}; our learned inhibitory conductances provide a trainable balancing mechanism.
    The same denominator also gives the direct connection to path gain. Writing $G_n^E(x)$ and $G_n^I(x)$ for the total excitatory and inhibitory synaptic conductance impinging on compartment $n$, and $G_n^{\mathrm{den}}=\sum_jg_j^{\mathrm{den}}$ for its total child-coupling conductance,
    \[
    R_n^{\mathrm{tot}}(x)=\left(1+G_n^E(x)+G_n^I(x)+G_n^{\mathrm{den}}\right)^{-1}.
    \]
    Increasing $G_n^I(x)$ lowers $R_n^{\mathrm{tot}}(x)$ and therefore multiplicatively suppresses every upstream path whose error must pass through compartment $n$.
    In this sense, inhibition can gate credit flow by changing the path gain, even when the inhibitory drive is feedforward rather than lateral. Prop.~\ref{prop:inhibitory_path_gain} formalizes this statement after the tree path gain is defined.
    
    \noindent\emph{Exact gradients for dendritic trees.}
    For one neuron/tree in layer $\ell$, let $V_{0,u}^{(\ell)}$ be the somatic voltage and define the exact voltage-space soma error as $\delta_{0,u}^{(\ell),V}:=\partial L/\partial V_{0,u}^{(\ell)}$.
    If a decoder receives soma activity $a_0=f_0(V_0)$ rather than raw voltage, the decoder supplies an activation-space error $\delta_0^a$ and the voltage-space boundary is $\delta_0^V=f_0'(V_0)\delta_0^a$.
    The theoretical derivation begins from the conductance-stage steady-state voltage. A branch may then apply a monotone secondary nonlinearity, which we call a reactivation, $a_n=f_n(V_n)$, before transmitting activity upward; the identity convention above covers the disabled case. All exact-error diagnostics below are computed as pre-reactivation voltage errors and therefore include the local derivatives $f_n'(V_n)$ along the path.
    
    \begin{theorem}[Conductance-Stage Backpropagation on a Dendritic Tree]
    \label{thm:tree_backprop}
    For a rooted dendritic tree with soma at node $0$ and unique parent $p(n)$ for each non-somatic compartment, the parent receives the reactivated activity $a_n=f_n(V_n)$, with $f_n$ equal to the identity when reactivation is disabled. The loss gradient satisfies
    \begin{equation}
    \frac{\partial L}{\partial V_n}
    = \frac{\partial L}{\partial V_{p(n)}}\, f_n'(V_n)R_{p(n)}^{\mathrm{tot}}\, g_{n\to p(n)}^{\mathrm{den}},
    \label{eq:tree_recursion}
    \end{equation}
    with boundary condition $\partial L / \partial V_0 = \delta_0^V$. Because the graph is a tree, each compartment has a unique path to the soma. We define the conductance-only path gain
    \begin{equation}
    \alpha_n^{\mathrm{cond}}
    = \prod_{(i\to k)\in \mathrm{path}(n\leadsto 0)} R_k^{\mathrm{tot}}\, g_{i\to k}^{\mathrm{den}},
    \label{eq:path_sum}
    \end{equation}
    and the exact effective path gain
    \begin{equation}
    \tilde{\alpha}_n
    = \prod_{(i\to k)\in \mathrm{path}(n\leadsto 0)}
    f_i'(V_i)\,R_k^{\mathrm{tot}}\,g_{i\to k}^{\mathrm{den}},
    \qquad
    \frac{\partial L}{\partial V_n} = \tilde{\alpha}_n\delta_0^V,
    \label{eq:effective_path_gain}
    \end{equation}
    with $\alpha_0^{\mathrm{cond}}=\tilde{\alpha}_0=1$. When reactivation is disabled, $f_i'(V_i)=1$ along the path and $\tilde{\alpha}_n=\alpha_n^{\mathrm{cond}}$.
    \end{theorem}
    The path gain is not an additional parameter that the local rule must learn. It is the derivative induced by the current dendritic coupling, input resistance, and activation state. Exact transport diagnostics compute this derivative as an oracle; the restricted broadcast rules deliberately do not estimate it and therefore expose the cost of omitting path-specific transport.
    \begin{proof}
    Apply the chain rule on the tree-structured computation graph using Prop.~\ref{prop:sensitivities}.
    \end{proof}
    The one-step recursion in Eq.~\eqref{eq:tree_recursion} and the effective path product in Eq.~\eqref{eq:effective_path_gain} separate local child-to-parent transfer from the non-local somatic error for one soma/tree and one layer. Supplying an approximate layer-soma teaching vector, reusing a decoder-derived vector at another layer, and compressing $\delta_{0,u}^{(\ell),V}$ into an ancestry-shared or global scalar broadcast are additional potential feedback-approximation steps introduced by LocalCA, not consequences of the theorem.

    \begin{corollary}[Local--Global Factorization]
    \label{cor:factorization}
    The exact synaptic gradient at compartment $n$ factorizes as:
    \begin{equation}
    \frac{\partial L}{\partial g_i^{\mathrm{syn}}} =
    \underbrace{x_i\, R_n^{\mathrm{tot}}\, (E_i - V_n)}_{\text{synapse-local eligibility}} \;\cdot\;
    \underbrace{\frac{\partial L}{\partial V_n}}_{\text{compartment error}},
    \label{eq:factorization}
    \end{equation}
    where the eligibility term depends only on quantities available at the synapse $i$ (presynaptic activity $x_i$, input resistance $R_n^{\mathrm{tot}}$, driving force $E_i - V_n$), and the compartment error $\partial L / \partial V_n$ is the sole non-local quantity.
    \end{corollary}
    Exact-gradient reconstruction uses the measured $R_n^{\mathrm{tot}}$ from the forward pass. A biological implementation need not explicitly represent this factor with numerical precision: if branch resistance varies slowly or only by small amounts, it can be approximated or absorbed into branch-specific plasticity gain, and it provides a natural site for neuromodulatory regulation of learning rate.
    
    \begin{proposition}[Inhibitory control of conductance-stage path gain]
    \label{prop:inhibitory_path_gain}
    Let $\mathcal{A}(n)$ be the set of compartments on the path from compartment $n$ to the soma, excluding $n$ and including the parent compartments through which the error must pass. Holding dendritic coupling conductances fixed, an added inhibitory conductance $\Delta G_k^I(x)\geq 0$ at any $k\in\mathcal{A}(n)$ changes the conductance-stage path gain by
    \begin{equation}
    \frac{\alpha_n^{\mathrm{cond}}(x;\Delta G^I)}{\alpha_n^{\mathrm{cond}}(x;0)}
    =
    \prod_{k\in\mathcal{A}(n)}
    \frac{g_k^{\mathrm{tot}}(x)}
    {g_k^{\mathrm{tot}}(x)+\Delta G_k^I(x)} ,
    \qquad
    \frac{\partial \log \alpha_n^{\mathrm{cond}}}{\partial G_k^I} =
    \begin{cases}
    -R_k^{\mathrm{tot}}, & k\in\mathcal{A}(n),\\
    0, & k\notin\mathcal{A}(n).
    \end{cases}
    \label{eq:inhibition_path_gain_ratio}
    \end{equation}
    \end{proposition}
    \begin{proof}
    Substitute $R_k^{\mathrm{tot}}=1/g_k^{\mathrm{tot}}$ into the path-gain product in Eq.~\eqref{eq:path_sum}. Adding inhibitory conductance changes only the denominator at compartments on the path from $n$ to the soma, which gives the product ratio and the log derivative in Eq.~\eqref{eq:inhibition_path_gain_ratio}.
    \end{proof}
    
    Prop.~\ref{prop:inhibitory_path_gain} makes two directed predictions. First, inhibition is a path gate: it attenuates all credit paths below the inhibited compartment. Second, inhibition is useful for restricted-feedback local learning only when this attenuation makes the distribution of $\alpha_n^{\mathrm{cond}}$ more uniform or better aligned with the task. If inhibition is too strong, too heterogeneous, or generated by a poorly calibrated population of inhibitory synapses, it can suppress useful signals and hurt learning. In the additive case, the same inhibitory input changes voltages and local eligibilities, but it does not enter the conductance-stage path recursion as a multiplicative resistance gate.
    
    \noindent\textbf{When does shunting compress path gains?}
    Attenuation is not by itself a concentration result. Writing the log path gain as $z$ and the path-dependent inhibitory attenuation as $\beta$, the perturbed field is $z'=z-\beta$ and becomes less dispersed only when inhibition preferentially attenuates initially high-gain paths. Appendix Eq.~\eqref{eq:compression_condition} gives the exact covariance condition and explains why it need not agree with coefficient-of-variation or error-rank summaries that average over different axes. We therefore test path-gain dispersion, exact-error geometry, and inhibition interventions directly rather than treating compression as a theorem.
    
    Conditional on an exact layer-soma error, the exact and 3F within-tree gradients differ only by the effective path gain. Appendix Prop.~\ref{prop:path_gain_alignment} shows that their per-example cosine decreases with the eligibility-weighted dispersion of that gain. This explains why a more uniform gain field is useful, but does not solve the inter-layer problem: practical LocalCA must first approximate the correct layer-soma error. The implemented rules are batch averages, and the 4F/5F factors introduced below are slow empirical preconditioners rather than additional error signals. We therefore measure compatibility with the complete compartment-error field directly.
    Whether shunting improves this approximation is an empirical question; the matched gradient, activation, and backward-only diagnostics below separate the supported full-update effect from stronger causal interpretations.
    In the empirical networks, each dendritic branch layer may apply a reactivation $a_n=f_n(V_n)$ before passing activity to the next stage; unless otherwise noted, the manuscript sweeps use a learnable bounded tanh reactivation.
    
    \section{Local Learning Rules}
    \label{sec:local_rules}
    
    \noindent\emph{Broadcast error approximation.}
    We approximate the exact compartment error $\partial L / \partial V_n$ (Corollary~\ref{cor:factorization}) with a broadcast signal $e_n^V$ derived from a soma/core teaching vector. For a linear readout this vector is $W_{\mathrm{dec}}^\top (\partial L / \partial \hat{y})$; for a nonlinear decoder we use the corresponding decoder-input Jacobian product $J_{\mathrm{dec}}(h_{\mathrm{core}})^\top (\partial L / \partial \hat{y})$ so that feedback is defined in soma/core coordinates.
    If this vector is used at an earlier matched-width dendritic layer, practical LocalCA assumes neuron-index correspondence and reuses the coordinate as an approximate layer-soma error. This inter-layer teaching approximation is distinct from the within-tree question addressed by Theorem~\ref{thm:tree_backprop}: given a layer-soma error, how faithfully can a restricted branch-level field approximate $\partial L/\partial V_n$?
    The main configurations use a \textbf{matched-width signal with scalar fallback}. If a dendritic stage has the same width as the soma/core teaching vector, that vector is reused coordinate-wise; at wider branch stages, the vector is reduced to one scalar per example and broadcast to every compartment. This is not a per-soma signal because a soma coordinate is not repeated over all of its descendants. The \textbf{ancestry-shared soma signal} performs that operation explicitly: for a soma $u$, its teaching coordinate is repeated over all branch compartments descended from $u$. This uses one teaching coordinate per neuron rather than one independently specified error per compartment; dendritic ancestry supplies the routing map without dense cross-neuron feedback. Conditional on an exact layer-soma error, this mode isolates the within-tree sharing approximation. Practical multi-layer LocalCA can additionally reuse final-core coordinates as approximate earlier-layer soma errors, as separated in Fig.~\ref{fig:gradient_fidelity}D.

    A \textbf{global scalar} compresses the soma/core vector to one number per example, $\bar{\delta}_b=d_0^{-1}\sum_{c=1}^{d_0}\delta_{0,bc}^{a}$, and broadcasts that number to all neurons and compartments, where $d_0$ is the soma/core teaching-vector dimension and $\mathbf{1}$ is an all-ones vector over the receiving stage. For higher-bandwidth controls, $P_K$ is a fixed random projection from soma-error coordinates to $K$ feedback channels, $q_{n,k}$ are fixed compartment-specific mixing coefficients, and $\Gamma_n$ denotes the pathway-structured transport map for compartment $n$. We therefore distinguish six feedback objects: \textbf{global scalar}, $e_n=\bar{\delta}\mathbf{1}$; \textbf{matched-width/scalar-fallback}, the hybrid just defined; \textbf{ancestry-shared soma}, $e_{u,n}$ is one coordinate repeated over descendants of soma $u$; \textbf{rank-$K$ random}, $c=P_K\delta_0^a$ and $e_n=\sum_k q_{n,k}c_k$; \textbf{path-structured}, $e_n=\lambda_{\mathrm{mix}}\bar{\delta}\mathbf{1}+(1-\lambda_{\mathrm{mix}})\Gamma_n(\delta_0^a)$ with mixing weight $\lambda_{\mathrm{mix}}\in[0,1]$; and the \textbf{transported oracle}, $e_n^{\mathrm{transport}}=\tilde{\alpha}_n\delta_0^V=\partial L/\partial V_n$. The oracle is an analysis upper bound because it supplies the exact path-transported error from Theorem~\ref{thm:tree_backprop}. When reactivation is enabled, exact path transport includes activation derivatives along the branch-to-soma path. Appendix Table~\ref{tab:effective_broadcast_dimensionality} summarizes the dimensionality and spatial constraints of these objects.

    The main matched-width/scalar-fallback condition is more restrictive than a true ancestry-shared soma broadcast at distal and proximal branch stages, but less restrictive at its matched-width stage than one layer-wide scalar. With shunting inhibition, this hybrid field supports nontrivial learning on standard classification tasks. Scalar, ancestry-shared, routed, low-rank, and path-transport controls expose which losses arise from feedback dimensionality, ancestry, inter-layer reuse, and within-tree transport. A compartment-local mismatch is retained as a negative control because it performs much worse than the somatic-error broadcasts (Fig.~\ref{fig:rule_feedback_controls}B).
    
    \noindent\emph{Three-factor learning rule (3F).}
    \begin{definition}[3F Update]
    For synaptic and dendritic conductances, using $a_j=f_j(V_j)$ and the identity convention when reactivation is disabled:
    \begin{equation}
    \widehat{\nabla}_{g_j^{\mathrm{syn}}}L = \langle x_j R_n^{\mathrm{tot}} (E_j - V_n)\, e_n^V \rangle_B,
    \qquad
    \widehat{\nabla}_{g_j^{\mathrm{den}}}L = \langle R_n^{\mathrm{tot}} (a_j - V_n)\, e_n^V \rangle_B,
    \label{eq:3f}
    \end{equation}
    where $\langle\cdot\rangle_B$ denotes the batch average.
    \end{definition}
    
    The three factors are presynaptic activity $x_j$ (or a transmitted child activity difference), postsynaptic modulation through driving force and input resistance, and the broadcast voltage-space error $e_n^V$.
    The same rule applies to excitatory and inhibitory synapses. The sign difference comes from the driving force $(E_j - V_n)$.
    We write $\widehat{\nabla}_g L$ as the local gradient estimate supplied to the optimizer; the optimizer applies the usual learning rate and descent step. Equivalently, one could absorb the sign convention into $e_n^V$ and treat it as a descent signal.
    
    \noindent\textbf{Additive control.}
    Rule~\eqref{eq:3f} is the local gradient of the \emph{shunting} voltage $V_n = \sum E_j g_j x_j / g_n^{\mathrm{tot}}$.
    For the additive E/I control, inhibition enters as a fixed signed voltage contribution rather than through the conductance denominator:
    \[
    V_n^{\mathrm{add}}=\sum_{j\in E}g_j^E x_j^E-\sum_{j\in I}g_j^I x_j^I+\sum_k g_k^{\mathrm{den}}a_k .
    \]
    Equivalently, $\partial V_n^{\mathrm{add}}/\partial g_j^{\mathrm{syn}}=s_jx_j$ with $s_j=+1$ for excitatory and $s_j=-1$ for inhibitory synapses, and $\partial V_n^{\mathrm{add}}/\partial g_k^{\mathrm{den}}=a_k$. The additive local gradient therefore has no driving-force or $R_n^{\mathrm{tot}}$ terms ($R_n^{\mathrm{tot}}\equiv1$ by definition since there is no denominator).
    For the additive tree, the conductance-only upward gain is $\alpha_n^{\mathrm{add}}=\prod_{(i\to k)\in\mathrm{path}(n\leadsto0)}g_{i\to k}^{\mathrm{den}}$ and the exact effective gain is $\tilde{\alpha}_n^{\mathrm{add}}=\prod_{(i\to k)}f_i'(V_i)g_{i\to k}^{\mathrm{den}}$; these coincide when reactivation is disabled. Thus additive controls have an exact tree gain, but inhibition does not enter that gain through a conductance denominator.
    Throughout, each architecture uses the learning rule derived from its own forward-pass dynamics. This keeps the comparison architecture-matched rather than applying a shunting-derived rule to an additive model; Appendix Table~\ref{tab:shunting_additive_mechanism} lists the corresponding forward and local-gradient terms.
    Let $\theta$ denote an unconstrained optimizer parameter and $g=\mathrm{softplus}(\theta)$ its nonnegative physical conductance. Eqs.~\eqref{eq:3f}--\eqref{eq:5f} give conductance-space gradients; before assigning a gradient to $\theta$, we multiply by $\partial g/\partial\theta$. The exact-gradient reconstruction diagnostics include this parameterization factor; Appendix~\ref{app:implementation} gives the compact shunting/additive comparison and Appendix Table~\ref{tab:eq_code_map} maps the equations to their implementation.
    
    \noindent\emph{Practical heuristic wrappers: 4F and 5F.}
    The 4F and 5F variants add slow branch-level preconditioners to the same fast eligibility and broadcast error; they are empirical stabilizers, not additional task-error channels. 4F uses a branch-soma covariance proxy and 5F adds a bounded predictability factor,
    \[
    r_n^{\mathrm{4F}}=
    \frac{\Cov_B(\bar V_{n,b},\bar V_{0,b})}
    {\sqrt{\Var_B(\bar V_{n,b})\Var_B(\bar V_{0,b})}+\varepsilon},
    \qquad
    \phi_n=\frac{\Var(V_n)}{\sigma_{\mathrm{res},n}^2+\varepsilon},
    \]
    where $\bar V_{n,b}$ and $\bar V_{0,b}$ are scalar summaries of branch and soma activity for example $b$, $\Cov_B$ and $\Var_B$ are computed across batch examples, $\sigma_{\mathrm{res},n}^2$ is the online residual variance from predicting branch activity with its configured parent or soma proxy, and $\varepsilon>0$ stabilizes the denominators. The numerator $\Var(V_n)$ uses the same configured online variance estimator. The implementation clamps $r_n^{\mathrm{4F}}\in[0.1,2.0]$ and $\phi_n\in[0.25,4.0]$. The appendix specifies the averaging axes, EMA initialization, detaching, and clamp rationale. Empirically, 4F alone gives little improvement over 3F; the practical rule is the 5F update below.
    
    \begin{definition}[5F Update]
    \begin{equation}
    \widehat{\nabla}_{g_j^{\mathrm{syn}}}L = r_n^{\mathrm{4F}} \phi_n \langle x_j R_n^{\mathrm{tot}} (E_j - V_n)\, e_n^V \rangle_B,
    \qquad
    \widehat{\nabla}_{g_j^{\mathrm{den}}}L = r_n^{\mathrm{4F}} \phi_n \langle R_n^{\mathrm{tot}} (a_j - V_n)\, e_n^V \rangle_B.
    \label{eq:5f}
    \end{equation}
    \end{definition}
    Eq.~\eqref{eq:5f} is therefore the main practical LocalCA update; Algorithm~\ref{alg:main_localca_update} (Appendix~\ref{app:implementation}) gives the full forward-pass-to-weight-update order. The 3F rule uses the theorem-derived local eligibility with an approximate broadcast error, while 4F and 5F add empirical reliability factors.
    A feedback-alignment-style random-broadcast argument is included in Appendix~\ref{app:theory}; it is supportive but not central to the main mechanism claim.
    
    \section{Experiments}
    \label{sec:experiments}
    
    \begin{figure}[t]
    \centering
    \includegraphics[width=\textwidth]{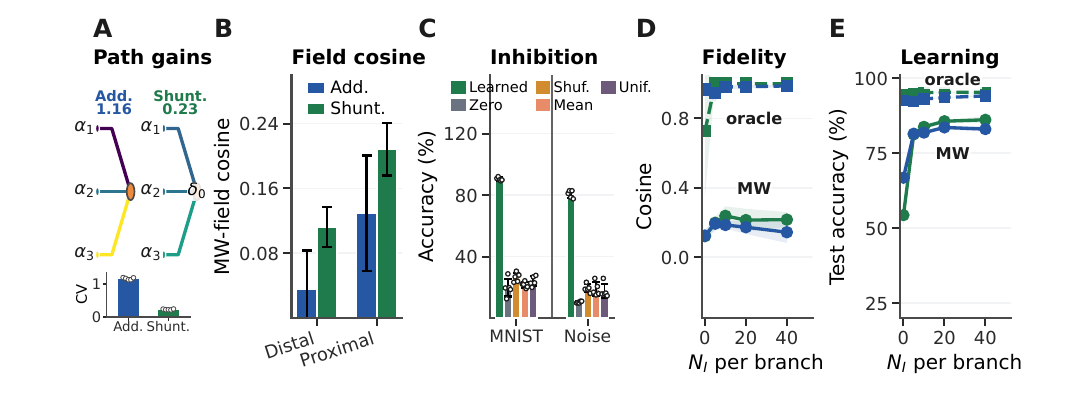}
    \caption{\textbf{Mechanistic chain from path gains to learning.}
    (A) Conductance-stage path-gain field $\alpha^{\mathrm{cond}}$ at $N_I{=}5$; the inset shows the five paired-seed CV values, with lower shunting CV in every pair. (B) Stage-resolved cosine between the matched-width/scalar-fallback field and the exact MNIST error on distal and proximal compartments, averaged over train, validation, and test batches within each of five checkpoints. The width-matched somatic stage is excluded because its cosine is one by construction in both cores. Fidelity is low on both dendritic stages, and the five-checkpoint comparison is descriptive rather than evidence for a cross-core ordering. (C) Post-training intervention in the learned shunting model: inhibition is zeroed, sample-shuffled, replaced by a per-branch batch mean, or replaced by one uniform matched mean. (D,E) Fidelity and learning on noise resilience; green denotes shunting, blue additive, solid lines matched-width/scalar-fallback feedback (MW), and dashed lines exact transported error (oracle). CV denotes the coefficient of variation over the sampled path-gain field and $N_I$ the number of inhibitory synapses per branch. Metrics and interventions: Appendix~\ref{app:extra_results}. Error bars are $\pm$1 s.d. over checkpoints.}
    \label{fig:mechanistic}
    \end{figure}
    
    \begin{figure}[t]
    \centering
    \includegraphics[width=\textwidth]{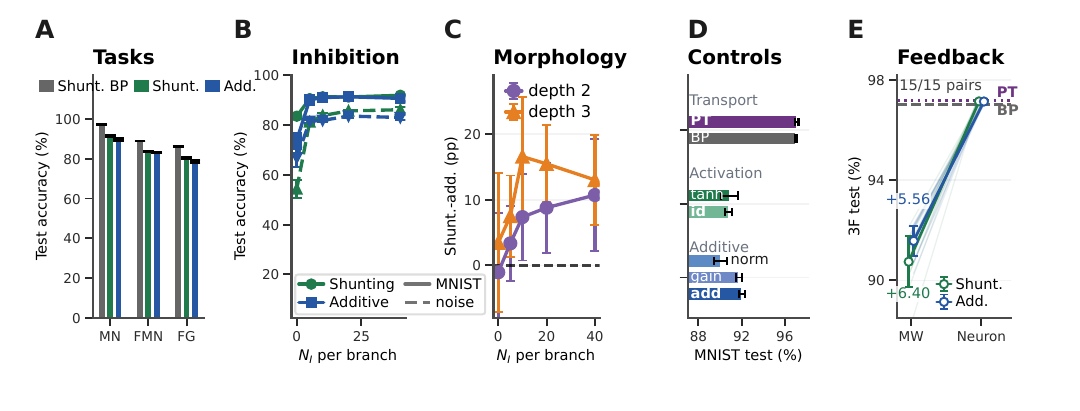}
    \caption{\textbf{Matched-capacity performance, regime dependence, and feedback definition.}
    (A) Matched shunting backpropagation reference (grey, labelled Shunt.\ BP) vs.\ 5F LocalCA with matched-width/scalar-fallback feedback; MN, FMN, and FG denote MNIST, Fashion-MNIST, and figure-ground MNIST. Additive backpropagation references are reported in Table~\ref{tab:gap_closing} where matched five-seed runs are available. (B) Inhibitory dose-response. (C) Descriptive morphology-by-inhibition performance; the detailed map and its initialization limitation are in Appendix Fig.~\ref{fig:morphology_ie_regime}. (D) MNIST controls: BP and exact path transport (PT); learned tanh versus identity transfer (id.); normalized eligibility (+norm); gain-matched additive eligibility (gain); and the additive core (add.). (E) Fifteen-seed 3F intervention: replacing matched-width/scalar-fallback feedback (MW) by one neuron-indexed coordinate shared within each modeled tree raises accuracy by $6.40$ points in shunting and $5.56$ points in additive networks; all paired seeds improve. Dashed and dotted lines mark the matched shunting BP and exact-transport references. Values/protocols: Table~\ref{tab:gap_closing}; Appendix~\ref{app:extra_results}. Error bars denote $\pm$1 s.d.}
    \label{fig:competence_regime}
    \end{figure}
    
    \noindent\emph{Setup.}
    We use architecture-matched additive/shunting cores with local or backprop training on MNIST \cite{lecun1998gradient}, Fashion-MNIST \cite{xiao2017fashionmnist}, figure-ground MNIST, and noise resilience. CIFAR-10 \cite{krizhevsky2009learning}, cue routing, and DFA are stress or boundary controls. Unless stated otherwise, local runs use 5F with matched-width/scalar-fallback feedback and main results use five seeds. Transported error is an oracle, not a proposed biological rule. Appendices~\ref{app:extra_results}--\ref{app:synthetic_tasks} and Tables~\ref{tab:reproducibility_summary}--\ref{tab:compute_resources} specify metrics, implementation, tasks, seeds, settings, and compute.
    
    \noindent\emph{From exact factorization to credit-signal fidelity.}
    Exact-gradient reconstruction matches autograd. At non-somatic parameter blocks, an archived three-checkpoint cohort favors shunting in direction and scale, but a five-seed replication does not reproduce the directional contrast. Whole-model concatenation reverses because an aligned soma block carries different gradient energy, and an identity-transfer control finds better shunting direction but slightly better additive learning. We therefore treat cross-core direction as sensitive to activation and initialization conditions rather than as an architecture-wide advantage (Fig.~\ref{fig:gradient_fidelity}A--C; Appendix Fig.~\ref{fig:mechanistic_audit}; Tables~\ref{tab:matched_3f_gradient}--\ref{tab:identity_transfer_3f}).

    Feedback controls are decisive. Exact layer-soma errors plus path transport reconstruct autograd, whereas neuron-wise sharing omits path gains and the main hybrid also collapses wider stages to a scalar (Fig.~\ref{fig:gradient_fidelity}D). Exact transport brings 3F and 5F to the backprop ceiling; replacing scalar fallback by one soma coordinate per neuron's compartments closes most of the 3F MNIST gap in both cores (Figs.~\ref{fig:competence_regime}E and \ref{fig:rule_feedback_controls}; Appendix Tables~\ref{tab:exact_transport_factorial} and \ref{tab:feedback_definition_control}). A compartment-local mismatch fails as a teaching signal, while the noise-resilience ladder shows that additional bandwidth or path information can help without guaranteeing monotonic improvement (Fig.~\ref{fig:rule_feedback_controls}). Since current aliases use different gate-initialization policies, only the paired feedback substitution within each core is causal. An explicit 2x2 factorial finds equal 3F accuracy under analytical initialization and a modest additive advantage under occupancy calibration, confirming that the cross-core sign is policy-dependent (Appendix Table~\ref{tab:init_policy_factorial}). This regular-tree control identifies collapsed neuron identity as a major bottleneck, not realistic morphology as an optimal router.
    
    \noindent\emph{Mechanistic chain: path gains, broadcast fidelity, and oracle transport.}
    At matched inhibition, shunting narrows the conductance-stage path-gain distribution: at $N_I{=}5$, its CV is about $80\%$ lower than additive, with the same ordering in all five pairs and at every nonzero inhibitory count (Fig.~\ref{fig:mechanistic}A). This unweighted statistic is not the full activation-dependent gain or gradient, so lower CV does not imply better direction. A stage-resolved analysis also shows that the matched-width/scalar-fallback field is poorly aligned with exact errors on distal and proximal compartments in both cores (Fig.~\ref{fig:mechanistic}B). We therefore do not use pooled error-field geometry or fidelity as evidence for a shunting-specific compatibility advantage; unrestricted SVD is retained only as a descriptive geometry diagnostic (Appendix Fig.~\ref{fig:inhibition_causality_error_rank}).
    
    Zeroing, shuffling, or clamping learned inhibition disrupts the trained computation, showing sample dependence and rejecting a matched global conductance load (Fig.~\ref{fig:mechanistic}C). Exact transport reconstructs compartment errors and largely closes the noise-resilience learning gap (Fig.~\ref{fig:mechanistic}D,E).
    
    \noindent\emph{Matched-capacity performance.}
    Under the main 5F rule with matched-width/scalar-fallback feedback, shunting LocalCA remains close to matched backprop on the three main classification tasks (Table~\ref{tab:gap_closing}; Fig.~\ref{fig:competence_regime}). This is not an intrinsic 3F limit: neuron-wise feedback closes the MNIST gap in both cores, and exact transport reaches the shunting reference with 3F or 5F and either decoder update. Activation and additive gain/normalization controls do not explain the main-condition contrast (Fig.~\ref{fig:competence_regime}D; Appendix Table~\ref{tab:exact_transport_factorial}; Fig.~\ref{fig:additive_norm_control}). Rule and sensitivity details are in Appendix Table~\ref{tab:matched_rule_comparison} and Figs.~\ref{fig:five_factor_sensitivity}--\ref{fig:verification_appendix}.
    
    \noindent\emph{Regime dependence of the shunting advantage.}
    With matched-width/scalar-fallback feedback, shunting is modestly better on MNIST, Fashion-MNIST, and figure-ground MNIST, and the gap grows on inhibition-sensitive noise resilience. The performance gap is not monotonic in $N_I$ or nominal tree depth (Fig.~\ref{fig:competence_regime}B,C; Appendix Fig.~\ref{fig:additional_stress_tests}). We treat the morphology grid as descriptive because analytical initialization places some shunting couplings at its $10^{-6}$ floor.

    \begin{figure}[htbp]
    \centering
    \includegraphics[width=\textwidth]{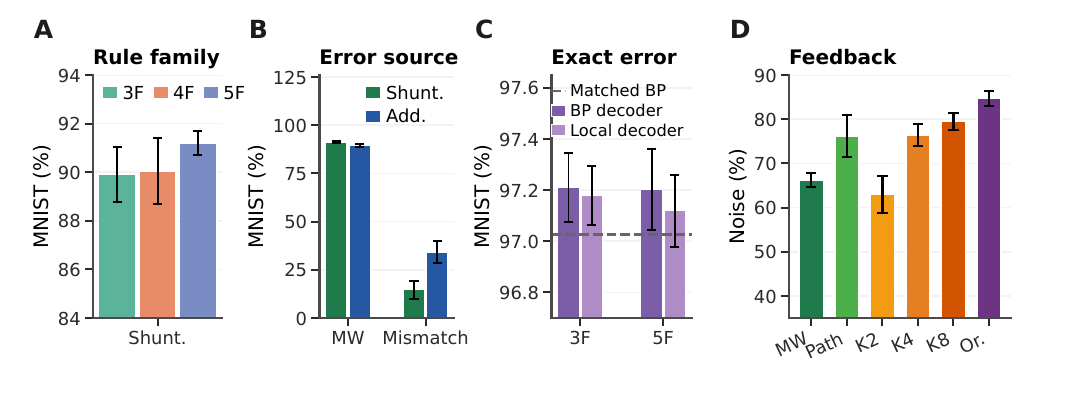}
    \caption{\textbf{Rule and feedback controls identify the principal learning bottleneck.}
    (A) Within-shunting comparison of 3F, 4F, and 5F in an archived three-seed cohort with matched feedback, decoder, optimizer, and effective parameter-group rates (descriptive). (B) Error-source negative control under the same local decoder: matched-width/scalar-fallback feedback (MW) supports learning, whereas a compartment-local mismatch does not ($n{=}3$); the cross-core ordering is not interpreted. (C) Five-seed exact-transport factorial: 3F and 5F both approach matched backpropagation with either a backpropagated or locally updated decoder; note the ordinate spans only $0.9$ percentage points. (D) Noise-resilience feedback ladder: MW, recursive Path, random low-rank K2/K4/K8, and transported Oracle. Together, the panels show that the feedback field, rather than a failure of the local eligibility factorization, sets the dominant practical limit. Error bars are $\pm1$ s.d.}
    \label{fig:rule_feedback_controls}
    \end{figure}
    
    \section{Discussion}
    \label{sec:discussion}
    The central contribution of this study is the exact separation of synapse-local eligibility---presynaptic activity, driving force, and input resistance---from path-specific compartment error. This turns LocalCA into a problem of approximating a defined credit field. The feedback analysis explains the weak matched-width/scalar-fallback 3F result: wider stages collapse neuron identity, whereas neuron-wise feedback nearly matches the performance of exact transport and backpropagation. This strengthens the local-eligibility result but does not establish a claim about realistic dendritic topology.

    Shunting is not universally beneficial, more inhibition is not always better, and 5F is an empirical preconditioner. Its directional contrast varies with activation and initialization conditions, although interventions show that learned inhibitory conductance is sample-dependent and necessary to the trained forward computation. Identity-transfer and fixed-state controls identify transfer derivatives, eligibility, and feedback identity---not isolated backward attenuation---as determinants of the final update (Appendix Fig.~\ref{fig:mechanistic_audit}).

    The model predicts that a focal conductance change should preferentially alter descendant plasticity signals, that current-matched shunting and additive perturbations should change teaching-signal gain differently, and that shared-ancestry branches should show more similar modulation than depth-matched branches from different subtrees. Testing these predictions requires simultaneous control of local input, voltage, and somatic teaching state (Appendix Table~\ref{tab:experimental_predictions}).

    A soma-level error could be carried by apical, burst, plateau, or neuromodulatory signals, but we do not identify a universal carrier. The requirement exposed here is a neuron-indexed teaching coordinate that influences dendritic plasticity; its biological generation remains open.

    The theorem assumes steady-state, non-spiking trees with one path per compartment. Active nonlinearities enter through local derivatives, whereas temporal calcium, spiking, recurrent paths, learned sparsity, and auxiliary objectives require extensions. The inhibition studied here is input-driven, and post-training interventions establish forward necessity rather than isolated credit-assignment causality. We therefore claim a credit-geometry mechanism supported by exact reconstruction, gradient diagnostics, feedback controls, and oracle transport---not competitive large-scale vision performance or sufficiency of one global scalar.

    \ifdefined\localcaauthorversion
    \section*{Acknowledgments}
    This work was supported in part by a gift from the Chan Zuckerberg Initiative Foundation to establish the Kempner Institute for the Study of Natural and Artificial Intelligence at Harvard University.
    \fi
    

    \appendix
    \renewcommand{\thefigure}{S\arabic{figure}}
    \setcounter{figure}{0}
    \renewcommand{\thetable}{S\arabic{table}}
    \setcounter{table}{0}
    
    \section{Supplementary Empirical Results}
    \label{app:extra_results}
    This appendix records the supporting diagnostics, implementation details, and task protocols needed to interpret the paper's main claims. It is organized by function: empirical diagnostics and boundary cases; implemented model and optimizer details; theory-adjacent extensions; and benchmark and synthetic-task protocols.

    \paragraph{Evidence map.}
    The exact-gradient and 3F claims are supported by the archived-checkpoint analysis, five-seed replication, occupancy, and backward-only controls below; the feedback-bottleneck claim by the exact-transport and neuron-wise factorials; the conductance interpretation by the inhibition intervention and operating-point analysis; and robustness by held-out seeds, stress tests, explicit inhibitory cells, and the CIFAR-10 control ladder. Main-text claims cite the corresponding supplementary figure or table directly.
    
    \paragraph{Scope notes.}
    Several controls are intentionally negative or mixed. The log-gain covariance margin from Eq.~\eqref{eq:compression_condition} is not used as a positive main result, and the pathway-vector feedback variant does not beat the best unstructured low-rank cue-routing control. These diagnostics define the scope of the main claims.

    \begin{table}[htbp]
    \centering
    \scriptsize
    \setlength{\tabcolsep}{2.5pt}
    \renewcommand{\arraystretch}{1.05}
    \begin{tabular}{@{}L{0.19\linewidth}L{0.25\linewidth}L{0.22\linewidth}L{0.27\linewidth}@{}}
    \toprule
    \textbf{Approach} & \textbf{Teaching or feedback object} & \textbf{Within-neuron biophysics} & \textbf{Relation to exact credit} \\
    \midrule
    Backpropagation \cite{rumelhart1986learning} & Exact reverse-mode error at every unit & Point units in the usual baseline & Computes the exact parameter gradient. \\
    FA / DFA \cite{lillicrap2016random,nokland2016dfa} & Fixed random backward or direct feedback & None required & Relies on learned alignment with an inexact feedback map. \\
    Target/control/predictive methods \cite{lee2015dtp,meulemans2021dfc,millidge2021predictive} & Learned targets, feedback control, or prediction errors & Circuit-level state, not a conductance tree & Reconstructs or approximates layer credit through auxiliary dynamics. \\
    Forward-only methods \cite{hinton2022forward,dellaferrera2022pepita} & Local goodness targets or output-error input perturbations & None required & Avoids an explicit exact reverse pass; does not derive compartment credit. \\
    Supervised local learning \cite{ma2024auglocal,lv2025dll,erdogan2025ebd,kao2024ccl} & Auxiliary local targets, broadcasts, or paired pathways & Dendritic structure in selected models & Optimizes practical local objectives rather than the exact conductance-tree error field. \\
    Dendritic teaching circuits \cite{guerguiev2017segregated,sacramento2018dendritic,payeur2021burst,greedy2022burstccn} & Segregated apical/basal, plateau, or burst signals & Compartmental circuit mechanisms & Constructs biologically motivated teaching signals that approximate backpropagation. \\
    \textbf{This work} & Restricted soma coordinates routed over a dendritic tree & E/I conductance, shunting, and morphology & Derives the exact conductance-tree credit object, then measures how well restricted feedback approximates it. \\
    \bottomrule
    \end{tabular}
    \caption{\textbf{Position relative to representative local-credit approaches.} The comparison separates the paper's mechanistic question---how within-neuron conductance and topology shape the exact required credit---from the complementary goal of designing scalable alternatives to backpropagation.}
    \label{tab:related_local_learning}
    \end{table}
    
    \begin{table}[htbp]
    \centering
    \scriptsize
    \setlength{\tabcolsep}{3pt}
    \renewcommand{\arraystretch}{1.06}
    \begin{tabular}{@{}L{0.19\linewidth}L{0.26\linewidth}L{0.47\linewidth}@{}}
    \toprule
    \textbf{Symbol} & \textbf{Name} & \textbf{Meaning in this manuscript} \\
    \midrule
    $V_n$, $a_n$ & Compartment voltage/activity & Pre-reactivation voltage and transmitted activity $a_n=f_n(V_n)$ for compartment $n$; $f_n$ is the identity when reactivation is disabled. \\
    $g^{\mathrm{syn}}$, $g^{\mathrm{den}}$ & Synaptic and dendritic conductances & Nonnegative conductance parameters for branch synapses and child-to-parent coupling. \\
    $R_n^{\mathrm{tot}}$ & Input resistance & Reciprocal of total conductance at compartment $n$; the implementation uses the stabilized $R_{n,\varepsilon}^{\mathrm{tot}}$. \\
    $E_j-V_n$ & Driving force & Difference between synaptic reversal potential and local voltage; multiplies presynaptic drive in conductance eligibility. \\
    $\delta_{0,u}^{(\ell),V}$, $\delta_{0,u}^{(\ell),a}$ & Soma voltage/activity error & Exact layer-soma error for neuron $u$ in layer $\ell$, before or after reactivation. \\
    $\delta_{u,n}^{(\ell)}$ & Compartment error & Exact pre-reactivation voltage-space error for compartment $n$ of neuron $u$ in layer $\ell$. \\
    $\boldsymbol{\delta}_0$, $\bar{\delta}$ & Broadcast source coordinates & Soma/core teaching vector and its global scalar average used by restricted-feedback controls. \\
    $\alpha_n^{\mathrm{cond}}$, $\tilde{\alpha}_n$ & Path gain & Conductance-stage and activation-derivative-corrected gain transporting soma error to compartment $n$. \\
    $e_n^V$ & Broadcast voltage-space error & LocalCA approximation to the exact compartment error used in the 3F--5F updates. \\
    $r_n^{\mathrm{4F}}$, $\phi_n$ & Branch-level preconditioners & Slowly estimated 4F covariance proxy and 5F bounded predictability factor. \\
    $\sigma_{\mathrm{res},n}^2$, $\varepsilon$ & Residual variance and stabilizer & Online branch-prediction residual variance and positive denominator stabilizer used by the 5F factor. \\
    $\theta$ & Raw optimizer parameter & Unconstrained variable mapped to a nonnegative conductance by $g=\mathrm{softplus}(\theta)$. \\
    $N_E$, $N_I$ & Synapse counts & Excitatory and inhibitory synapses per dendritic branch. \\
    \bottomrule
    \end{tabular}
    \caption{\textbf{Notation reference.} Main symbols used in the derivation, local-learning rules, and diagnostics.}
    \label{tab:notation_reference}
    \end{table}
    
    \subsection{Reproducibility Statement}
    
    The accompanying source package contains the model code, training scripts, diagnostic scripts, figure-generation scripts, representative configuration files, tests, and precomputed summary files used to generate the manuscript figures and tables. Table~\ref{tab:reproducibility_summary} summarizes the experiment families, seed protocols, and checkpoint-selection rules. Code and reproduction scripts will be released publicly after publication with an immutable repository tag, pinned environment file, end-to-end launch scripts, and expected summary hashes.
    
    \begin{table}[htbp]
    \centering
    \scriptsize
    \setlength{\tabcolsep}{3pt}
    \renewcommand{\arraystretch}{1.05}
    \begin{tabular}{@{}L{0.25\linewidth}L{0.27\linewidth}L{0.20\linewidth}L{0.20\linewidth}@{}}
    \toprule
    \textbf{Experiment family} & \textbf{Architecture / condition} & \textbf{Seeds} & \textbf{Selection / reporting} \\
    \midrule
    Exact-factorization and gradient diagnostics & Nonnegative-input E/I dendritic cores, $[128,128]$ soma layers, $[3,3]$ tree, 3F--5F local rules & Five seeds for main summaries; three seeds where indicated for gradient dynamics & Diagnostic batches or fixed checkpoints, reported as gradient/factorization diagnostics \\
    Layer-soma factorial diagnostic & Two-layer checkpoints with exact/reused soma errors and scalar-fallback, ancestry-shared, or transported branch fields & Six fixed checkpoints & Fixed-checkpoint diagnostic; no validation selection \\
    Fixed-state and norm-matched 3F analysis & Five-seed scalar-fallback MNIST checkpoints; factor substitutions and one-step branch updates & Five paired seeds & One held-out batch per checkpoint; no checkpoint reselection \\
    Identity-transfer 3F analysis & MNIST additive and shunting cores with unit branch-transfer derivative and matched-width/scalar-fallback feedback & Fifteen paired seeds & Validation-selected checkpoint; held-out accuracy and train/validation/test gradient diagnostics \\
    Feedback-definition control & MNIST, 3F, local decoder, matched-width/scalar-fallback vs.\ neuron-wise feedback within additive and shunting cores & Fifteen paired seeds per architecture and feedback mode & Validation-selected checkpoint; held-out test accuracy and fixed-checkpoint branch-gradient diagnostic \\
    Path-gain, exact-error rank, and inhibition intervention & Matched additive and shunting cores at fixed inhibition; post-training intervention states & Five seeds per main condition & Trained checkpoints; post-training interventions are evaluated on held-out data \\
    Transported-error oracle and feedback-construction controls & Matched-width/scalar-fallback, ancestry-shared, random low-rank, path-structured, and exact transported-error feedback & Five seeds for main oracle and feedback-ladder summaries & Validation-selected checkpoints; held-out test summaries \\
    Matched-capacity performance & MNIST, Fashion-MNIST, figure-ground MNIST; 5F LocalCA with matched-width/scalar-fallback feedback and matched backpropagation references & Five seeds & Validation-selected configurations; held-out test accuracy \\
    \bottomrule
    \end{tabular}
    \caption{\textbf{Reproducibility summary.} The table records the experiment families, model settings, seed counts, and selection rules used in the manuscript. Dependency specifications are provided through \texttt{requirements.txt}, \texttt{pyproject.toml}, and \texttt{setup/environment.yml}; exact end-to-end commands and summary hashes will be provided with the post-publication repository tag.}
    \label{tab:reproducibility_summary}
    \end{table}
    
    \subsection{Diagnostic Metrics}
    
    The gradient-fidelity diagnostics compare LocalCA and backpropagation gradients at matched parameters and batches. For a parameter block $p$, let $g_p^{\mathrm{local}}$ and $g_p^{\mathrm{bp}}$ be the two gradients and let $w_p=\mathrm{numel}(p)/\sum_q\mathrm{numel}(q)$. The reported weighted cosine is
    \[
    C_{\mathrm{w}}=\sum_p w_p
    \frac{\langle g_p^{\mathrm{local}},g_p^{\mathrm{bp}}\rangle}
    {\|g_p^{\mathrm{local}}\|_2\|g_p^{\mathrm{bp}}\|_2+\varepsilon},
    \]
    which is equivalent to concatenating the parameter blocks after separately normalizing each block to the same root-mean-square magnitude. It therefore measures typical per-parameter directional agreement without allowing one unusually high-energy block to determine the comparison. We also report the unweighted macro average across blocks and, separately, the raw whole-model concatenated cosine. The weighted scale mismatch is
    \[
    S_{\mathrm{w}}=\sum_p w_p
    \left|\log_{10}\frac{\|g_p^{\mathrm{local}}\|_2+\varepsilon}
    {\|g_p^{\mathrm{bp}}\|_2+\varepsilon}\right|.
    \]
    The checkpoint branch-scale panels instead use the single concatenated norm ratio
    \[
    Q_{\mathrm{branch}}=
    \frac{\sqrt{\sum_{p\in\mathcal{B}}\|g_p^{\mathrm{local}}\|_2^2}}
    {\sqrt{\sum_{p\in\mathcal{B}}\|g_p^{\mathrm{bp}}\|_2^2}},
    \]
    where $\mathcal{B}$ contains all non-somatic dendritic parameter blocks. Thus $Q_{\mathrm{branch}}=1$ denotes matched aggregate scale; it is distinct from the block-averaged log mismatch $S_{\mathrm{w}}$ used in the separate static diagnostic.
    For exact-error compressibility, we form a matrix $E$ whose rows are examples and whose columns are compartment--neuron coordinates of the exact pre-reactivation compartment error. If $E=U\Sigma V^\top$ and $\sigma_i$ are the singular values, the unrestricted rank-$1$ residual and participation rank are
    \[
    \rho_1=\frac{\sqrt{\sum_{i>1}\sigma_i^2}}{\sqrt{\sum_i\sigma_i^2}},
    \qquad
    r_{\mathrm{part}}=\frac{(\sum_i\sigma_i^2)^2}{\sum_i\sigma_i^4}.
    \]
    This is an intrinsic geometry metric: both the sample coefficients and spatial template are chosen in hindsight by SVD. The constrained shared-field residual instead fixes the spatial template to all ones and allows only the best per-example scalar,
    \[
    \rho_{\mathbf{1}}=
    \frac{\|E-\mathrm{rowmean}(E)\mathbf{1}^{\top}\|_F}{\|E\|_F}.
    \]
    The ancestry-shared/blockwise residual tests a different constrained family. Let $\Pi_{\mathrm{AS}}$ replace, for each sample and soma, all descendant compartments of that soma by their optimal shared coefficient. Then
    \[
    \rho_{\mathrm{AS}}=
    \frac{\|E-\Pi_{\mathrm{AS}}E\|_F}{\|E\|_F}.
    \]
    This is not the matched-width/scalar-fallback field at wider branch stages. The broadcast-conditioned residual uses the actual field $B$ generated by that mode,
    \[
    \rho_{\mathrm{bc}}=\frac{\|E-B\|_F}{\|E\|_F},
    \]
    with a corresponding flattened cosine and, where useful, a globally rescaled version of $B$. These constrained metrics are emitted by \path{scripts/measure_error_rank_diagnostics.py} alongside the unrestricted SVD metrics, so intrinsic error geometry and compatibility with the implemented broadcast field can be reported separately. Path-gain dispersion is the coefficient of variation of $\alpha_n^{\mathrm{cond}}$ or $\tilde{\alpha}_n$ over the measured sample and compartment field, depending on whether reactivation is enabled. This unweighted architectural statistic differs from the eligibility-weighted $\mathrm{CV}_w(\tilde{\alpha}_{n(i)})$ in Prop.~\ref{prop:path_gain_alignment}. Compartment-error fidelity is the cosine between the tested broadcast field and the exact compartment-error field after flattening over the same sample and compartment axes.
    
    \subsection{Matched Rule Comparison}
    The theorem-derived 3F eligibility is central to the proposed mechanism, whereas the main practical performance experiments use the empirical 5F stabilizer. Table~\ref{tab:matched_rule_comparison} reports an archived shunting cohort after matching feedback, decoder, optimizer, schedules, and effective parameter-group rates. Cross-core 3F conclusions instead use the explicit initialization-policy factorial below and the fifteen-seed feedback-definition control.

    \begin{table}[htbp]
    \centering
    \small
    \begin{tabular}{@{}lc@{}}
    \toprule
    \textbf{Rule} & \textbf{Shunting test accuracy} \\
    \midrule
    3F & $89.90\pm1.13$ \\
    4F & $90.05\pm1.35$ \\
    5F & $91.19\pm0.49$ \\
    \bottomrule
    \end{tabular}
    \caption{\textbf{Within-shunting rule comparison on MNIST.} Three unique seeds per rule from an archived cohort use matched-width/scalar-fallback feedback, a local decoder, Adam, and fixed effective parameter-group rates. The nominal base-LR sweep axis is inert because group-specific rates are fixed; duplicate nominal-LR rows yield identical checkpoints and are collapsed. The comparison is restricted to this matched shunting cohort and does not use the unmatched additive cohort.}
    \label{tab:matched_rule_comparison}
    \end{table}

    \begin{table}[htbp]
    \centering
    \small
    \begin{tabular}{@{}lcccc@{}}
    \toprule
    \textbf{Initialization policy} & \textbf{Shunting} & \textbf{Additive} &
    \textbf{Shunt.--add. (pp)} & \textbf{Paired $p$} \\
    \midrule
    Analytical & $90.92\pm0.42$ & $90.92\pm0.51$ & $-0.00$ & $0.982$ \\
    Occupancy quantile & $90.99\pm0.59$ & $91.66\pm0.79$ & $-0.67$ & $0.022$ \\
    \bottomrule
    \end{tabular}
    \caption{\textbf{Current-code 3F architecture-by-initialization factorial.}
    Fifteen paired seeds per cell use identical matched-width/scalar-fallback feedback, decoder, optimizer, and schedules. Occupancy calibration improves additive accuracy by $0.74$ percentage points but not shunting accuracy; the architecture-by-policy interaction has paired $t$ $p=0.036$ (signed-rank $p=0.041$). The occupancy procedure, including safety reversion when a fitted stage would exceed the predefined slope cap, is applied identically to both cores. Values are mean $\pm$ s.d.; the final column gives the paired $t$ test within policy.}
    \label{tab:init_policy_factorial}
    \end{table}

    \begin{table}[htbp]
    \centering
    \small
    \begin{tabular}{@{}lcc@{}}
    \toprule
    \textbf{Dendritic update} & \textbf{Decoder update} & \textbf{MNIST test acc.} \\
    \midrule
    Matched backpropagation reference & backpropagation & $97.03\pm0.08$ \\
    \midrule
    3F + exact transported error & backpropagation & $97.21\pm0.14$ \\
    3F + exact transported error & local & $97.18\pm0.12$ \\
    5F + exact transported error & backpropagation & $97.20\pm0.16$ \\
    5F + exact transported error & local & $97.12\pm0.14$ \\
    \bottomrule
    \end{tabular}
    \caption{\textbf{Exact-transport factorial.} Five matched shunting seeds per condition. Supplying the exact path-transported compartment error brings the theorem-derived 3F update to the matched backpropagation reference; neither the empirical 5F preconditioner nor a backpropagated decoder update provides a material additional benefit in this oracle condition. Exact transport is an analysis upper bound, not the proposed biological feedback rule.}
    \label{tab:exact_transport_factorial}
    \end{table}

    \begin{table}[htbp]
    \centering
    \small
    \begin{tabular}{@{}lcc@{}}
    \toprule
    \textbf{Fifteen-seed result} & \textbf{Shunting} & \textbf{Additive} \\
    \midrule
    Matched-width/scalar fallback & $90.75\pm1.01$ & $91.58\pm0.61$ \\
    Neuron-wise & $\mathbf{97.15\pm0.12}$ & $\mathbf{97.14\pm0.10}$ \\
    \midrule
    Paired accuracy gain (pp) & $+6.40$ & $+5.56$ \\
    Paired $t$ $p$ & $1.3{\times}10^{-12}$ & $1.3{\times}10^{-14}$ \\
    \midrule
    Branch cosine, scalar fallback & $0.007\pm0.068$ & $0.024\pm0.051$ \\
    Branch cosine, neuron-wise field & $\mathbf{0.667\pm0.043}$ & $\mathbf{0.599\pm0.034}$ \\
    \bottomrule
    \end{tabular}
    \caption{\textbf{Neuron-wise feedback closes the MNIST 3F gap.} Fifteen seeds are paired within each architecture. The neuron-wise mode repeats one soma coordinate over the 13 compartments of that modeled neuron's regular $[3,3]$ tree; it does not supply an independent error to every compartment. Accuracy rows compare separately trained feedback modes; cosine rows substitute scalar or neuron-wise diagnostic fields at the neuron-wise checkpoints. This tests the feedback implementation in the regular-tree credit-assignment model, not route selection in reconstructed morphologies. The architecture aliases retain their standard reactivation initialization policies, so feedback comparisons within a column are causal, whereas shunting--additive comparisons across columns are descriptive. Accuracy and cosine improvements occur in all 15 pairs (two-sided signed-rank $p=6.1{\times}10^{-5}$). Values are mean $\pm$ s.d.}
    \label{tab:feedback_definition_control}
    \end{table}
    
    \subsection{5F Stabilizer Sensitivity}
    Sensitivity checks show that 5F performance is not sharply dependent on the preconditioner clamp or the EMA rate used to estimate branch statistics. Across three-seed MNIST checks, tightening or widening the clamp and varying the EMA rate leaves test accuracy near the default setting, supporting the view that the factor acts as a bounded branch-level reliability preconditioner rather than a brittle tuned error source.
    
    \begin{figure}[htbp]
    \centering
    \includegraphics[width=\textwidth]{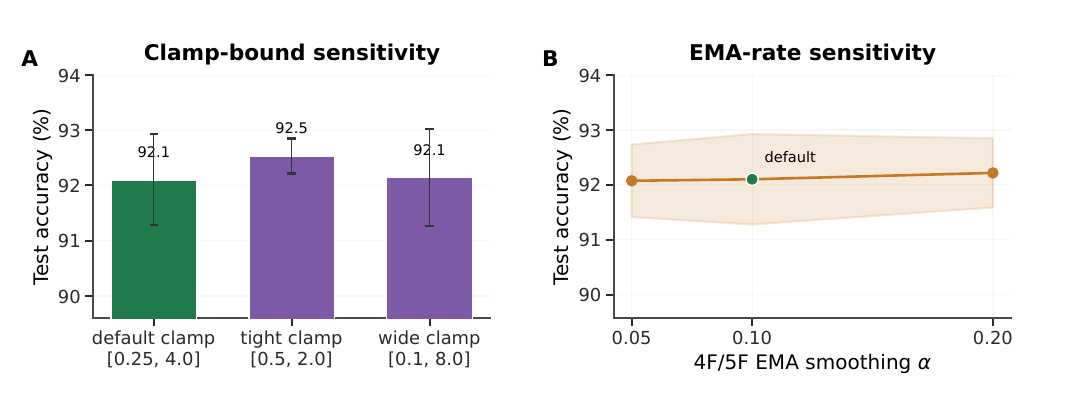}
    \caption{\textbf{5F stabilizer sensitivity.}
    \textbf{(A)}~Changing the 5F preconditioner clamp from the default $[0.25,4.0]$ to tighter $[0.5,2.0]$ or wider $[0.1,8.0]$ bounds leaves MNIST test accuracy near $92\%$.
    \textbf{(B)}~Changing the 4F/5F EMA smoothing rate over $\alpha\in\{0.05,0.10,0.20\}$ gives similarly stable performance. Error bars and bands are $\pm$1 s.d.\ across 3 seeds.}
    \label{fig:five_factor_sensitivity}
    \end{figure}
    \FloatBarrier
    
    \subsection{Local Performance and Regime Dependence}
    Table~\ref{tab:gap_closing} reports the matched-capacity performance numbers used in the main text. The point of this table is not to claim state-of-the-art benchmark performance, but to show the remaining LocalCA-to-backprop gap under matched dendritic capacity and the main 5F rule with matched-width/scalar-fallback feedback.
    
    \begin{table}[t]
    \centering
    \scriptsize
    \setlength{\tabcolsep}{3pt}
    \begin{tabular}{@{}lccccc@{}}
    \toprule
    \textbf{Dataset} & \textbf{Shunt. BP} & \textbf{Add. BP} & \textbf{Shunt. local} & \textbf{Add. local} & \textbf{Shunt. gap (pp)} \\
    \midrule
    MNIST & 0.971 & $0.970\pm0.001$ & $0.911\pm0.005$ & $0.894\pm0.007$ & 6.0 \\
    Fashion-MNIST & 0.889 & $0.876\pm0.001$ & $0.838\pm0.003$ & $0.831\pm0.004$ & 5.1 \\
    Figure-ground MNIST & 0.861 & -- & $0.803\pm0.006$ & $0.787\pm0.007$ & 5.8 \\
    \bottomrule
    \end{tabular}
    \caption{\textbf{Local performance.} Backpropagation references are matched-capacity standard-training runs. Additive backpropagation is reported where a matched five-seed reference is available; the figure-ground additive backpropagation reference was not available in the matched five-seed performance outputs. Local values use 5F with matched-width/scalar-fallback feedback and local decoder updates. The figure-ground LocalCA rows use the HSIC auxiliary objective (weight $0.01$); the shunting BP reference is the standard cross-entropy matched-capacity reference. Error bars are $\pm$1 s.d.\ across 5 seeds where available; gaps are percentage points.}
    \label{tab:gap_closing}
    \end{table}

    \subsection{Extended Gradient Analysis and $N_I$ Sweep Detail}
    Figure~\ref{fig:gradient_fidelity} in the main text summarizes final-state gradient diagnostics and the layer-soma factorial check. The supplementary view below (Fig.~\ref{fig:s2_gradient_extended}) separates static scale checks, inhibitory dose-response curves, and individual Fashion-MNIST seed behavior so that the aggregate claims are not driven by a single summary panel.
    
    \begin{figure}[htbp]
    \centering
    \includegraphics[width=\textwidth]{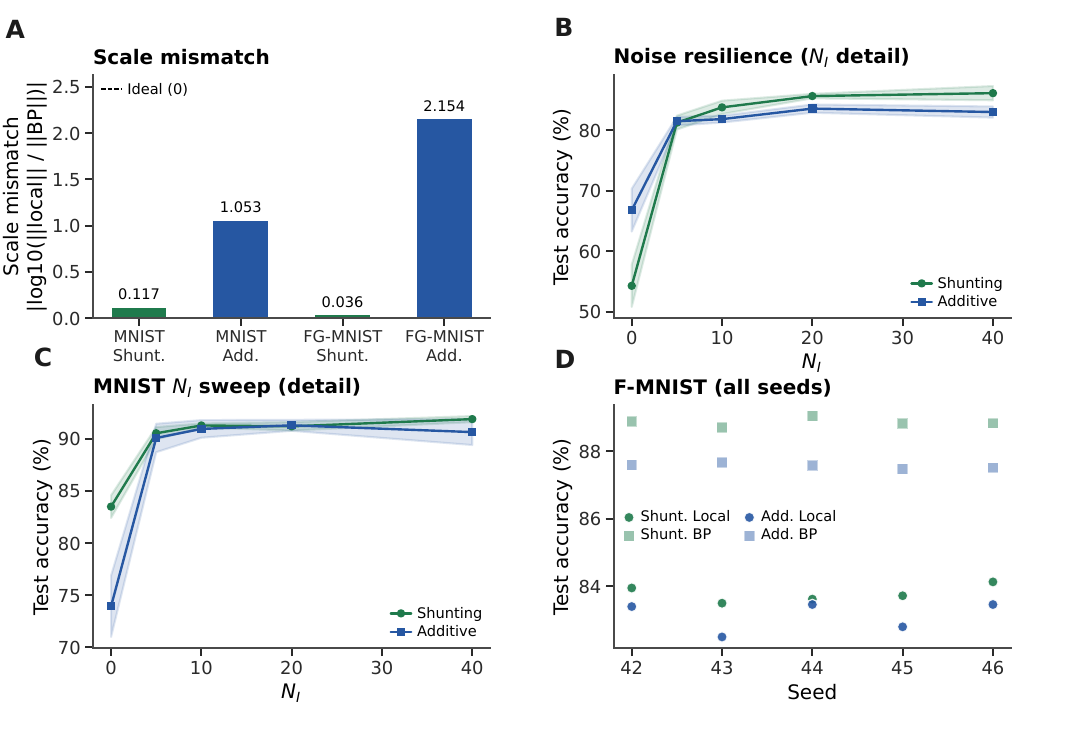}
    \caption{\textbf{Extended gradient and $N_I$ analysis (supplementary).} \textbf{(A)}~Log-scale mismatch, $|\log_{10}(\|g^{\mathrm{local}}\|/\|g^{\mathrm{bp}}\|)|$, from a separate matched-weight static diagnostic; these values should not be compared directly to the checkpointed final-state diagnostic in Fig.~\ref{fig:gradient_fidelity}. The ideal value is $0$; at this operating point shunting is near ideal ($0.117$), whereas additive shows order-of-magnitude distortion ($>1.0$). This descriptive static contrast is not evidence for a general cross-core advantage. \textbf{(B)}~Noise-resilience $N_I$ dose-response with error bands ($\pm$1 s.d.). \textbf{(C)}~MNIST $N_I$ dose-response detail with error bands. \textbf{(D)}~Fashion-MNIST individual seeds for all conditions, showing consistency across runs.}
    \label{fig:s2_gradient_extended}
    \end{figure}
    \FloatBarrier
    
    \paragraph{Alignment dynamics.}
    Local and backprop gradient norms stay finite throughout training, so weak additive alignment in this trajectory reflects directional and scale mismatch rather than vanishing gradients. In a separate archived three-seed checkpoint trajectory with parameter-level LocalCA and backprop gradients, additive local gradients remain nonzero (weighted norm range $9.5{\times}10^{-4}$--$9.5{\times}10^{-3}$), and the corresponding backprop gradients are also nonzero ($2.3{\times}10^{-5}$--$5.7{\times}10^{-3}$). The additive 5F cosine stays near zero from epoch $0$ to epoch $50$ ($0.018$ to $0.005$), while its local/backprop norm ratio falls from about $8.8{\times}10^2$ to $0.69$ (Fig.~\ref{fig:alignment_norm_dynamics}). These three archived checkpoints describe that operating point; the five-seed 3F and initialization-policy controls supersede them for any general cross-core directional claim.

    \begin{table}[htbp]
    \centering
    \scriptsize
    \begin{tabular}{@{}L{0.46\linewidth}cc@{}}
    \toprule
    \textbf{Matched 3F metric} & \textbf{Shunting} & \textbf{Additive} \\
    \midrule
    Whole-model concatenated cosine & $0.096\pm0.053$ & $0.350\pm0.059$ \\
    Branch parameter-count-weighted cosine & $0.160\pm0.007$ & $0.030\pm0.027$ \\
    Branch equal-block macro cosine & $0.238\pm0.067$ & $0.042\pm0.041$ \\
    Branch local/exact norm ratio & $0.261\pm0.035$ & $0.097\pm0.025$ \\
    Exact-gradient energy in final soma block & $0.2\%\pm0.03\%$ & $13.1\%\pm4.4\%$ \\
    \bottomrule
    \end{tabular}
    \caption{\textbf{Systematic matched 3F gradient diagnostic.} Three paired checkpoints, each averaged over the first train, validation, and test batches. The whole-model cosine answers a different, energy-allocation-dependent question because the final soma-coupling block is exactly aligned in both models. Branch metrics exclude this block and address the dendritic mechanism. Values are descriptive at $n=3$; inferential tests are omitted.}
    \label{tab:matched_3f_gradient}
    \end{table}

    \begin{table}[htbp]
    \centering
    \small
    \begin{tabular}{@{}lcc@{}}
    \toprule
    \textbf{Five-seed scalar-fallback 3F metric} & \textbf{Shunting} & \textbf{Additive} \\
    \midrule
    Branch parameter-count-weighted cosine & $0.155\pm0.045$ & $0.164\pm0.048$ \\
    Branch equal-block macro cosine & $0.166\pm0.041$ & $0.211\pm0.027$ \\
    Branch concatenated cosine & $0.143\pm0.049$ & $0.134\pm0.064$ \\
    Branch local/exact norm ratio & $0.313\pm0.046$ & $0.264\pm0.036$ \\
    \bottomrule
    \end{tabular}
    \caption{\textbf{Five-seed replication of the matched-width/scalar-fallback 3F diagnostic.} Each checkpoint is averaged over the first train, validation, and test diagnostic batches before seeds are paired. Direction does not differ reliably between architectures in this cohort, so the directional contrast in Table~\ref{tab:matched_3f_gradient} is treated as specific to the three archived checkpoint pairs. The core aliases retain their standard reactivation initialization policies, making this a replication of the operational training setup rather than an activation-matched causal intervention.}
    \label{tab:fresh_3f_replication}
    \end{table}

    \begin{table}[htbp]
    \centering
    \small
    \begin{tabular}{@{}lccc@{}}
    \toprule
    \textbf{Identity-transfer 3F result} & \textbf{Shunting} & \textbf{Additive} & \textbf{Paired $t$ $p$} \\
    \midrule
    Branch parameter-count-weighted cosine & $0.190\pm0.032$ & $0.111\pm0.023$ & $1.2{\times}10^{-6}$ \\
    Branch concatenated cosine & $0.152\pm0.024$ & $0.109\pm0.023$ & $1.0{\times}10^{-4}$ \\
    Branch local/exact norm ratio & $1.54\pm0.12$ & $0.078\pm0.014$ & $<10^{-10}$ \\
    MNIST test accuracy (\%) & $90.36\pm0.44$ & $90.93\pm0.18$ & $<0.001$ \\
    \bottomrule
    \end{tabular}
    \caption{\textbf{Derivative-matched identity-transfer control.} Fifteen paired seeds use matched-width/scalar-fallback 3F feedback, identical learning-rate schedules, and reactivation disabled, so $f'(V)=1$ in both architectures. Shunting has higher block-averaged and concatenated gradient direction, whereas additive has slightly higher accuracy. The update scales remain different, and identity transfer does not match forward voltage distributions. Cosine is invariant to global norm rescaling. Values are mean $\pm$ s.d.; paired $t$ tests are shown, and two-sided signed-rank tests give $p\leq0.0024$ for all four rows.}
    \label{tab:identity_transfer_3f}
    \end{table}

    \paragraph{Eligibility-weighted gradient and activation analysis.}
    A synaptic gradient weights compartment errors by sample-specific eligibility and maps them into parameter blocks. The analysis also finds different forward operating points. Mean activation derivatives across the distal, proximal, and soma stages are $(0.521,0.694,2.662)$ for shunting and $(0.061,0.057,0.257)$ for additive; the fractions with derivative magnitude below $0.1$ are $(0.680,0.459,0.237)$ and $(0.834,0.825,0.271)$, respectively. Learning-rate schedules are identical and cosine is invariant to global gradient rescaling, but these activation distributions are not matched. We therefore interpret the branch-gradient result as evidence for the full conductance-dependent update, not as an intervention isolating the backward path while holding the forward representation fixed. Recalibrating every trained gate to common voltage-quantile occupancy targets is a fixed-checkpoint sensitivity check on that interpretation rather than a learning comparison (Table~\ref{tab:occupancy_matched_3f}).

    \begin{table}[htbp]
    \centering
    \small
    \begin{tabular}{@{}lcc@{}}
    \toprule
    \textbf{Fixed-checkpoint 3F diagnostic} & \textbf{Shunting} & \textbf{Additive} \\
    \midrule
    Archived-state branch cosine & $0.160\pm0.007$ & $0.030\pm0.027$ \\
    Occupancy-calibrated branch cosine & $0.324\pm0.067$ & $0.141\pm0.021$ \\
    Occupancy-calibrated norm ratio & $1.01\pm0.16$ & $0.28\pm0.08$ \\
    \bottomrule
    \end{tabular}
    \caption{\textbf{Reactivation-occupancy sensitivity.} For the post-hoc intervention, every trained parametric gate is recalibrated so that its within-layer 10th and 90th voltage quantiles map to common 0.1 and 0.9 activation targets, after which exact and local gradients are recomputed on the same test batch. The result is a fixed-checkpoint sensitivity test, not a learning comparison: recalibration changes model predictions and does not make the full activation-derivative distributions identical. Three paired seeds; inferential values are descriptive.}
    \label{tab:occupancy_matched_3f}
    \end{table}

    \paragraph{Backward-only path-gain counterfactual.}
    To separate backward transport from the forward operating point, we recomputed an exact counterfactual target after removing inhibitory conductance only from the parent-resistance factors along the backward path. Recorded voltages, activation derivatives, dendritic couplings, soma errors, and all local eligibility variables were held fixed. We then compared the same restricted or neuron-wise 3F update with the actual and counterfactual targets. Removing backward attenuation leaves direction essentially unchanged even in the higher-inhibition noise regime (Table~\ref{tab:backward_only_3f}). Together with the measured conductances, this establishes an operating-point result: inhibition changes backward gain, but at the trained states its isolated contribution is too small to move credit direction. The complete conductance-dependent eligibility, feedback field, and operating state determine the update.

    \begin{table}[htbp]
    \centering
    \small
    \begin{tabular}{@{}L{0.25\linewidth}cL{0.21\linewidth}L{0.21\linewidth}@{}}
    \toprule
    \textbf{Fixed-checkpoint setting} & \textbf{Distal gain ratio} & \textbf{MW cosine: actual $\rightarrow$ no-I} & \textbf{Neuron-wise cosine: actual $\rightarrow$ no-I} \\
    \midrule
    MNIST & $1.070\pm0.008$ & $0.1422\rightarrow0.1426$ & $0.4957\rightarrow0.4968$ \\
    Noise, population 0 & $1.281\pm0.010$ & $0.0653\rightarrow0.0649$ & $0.1748\rightarrow0.1761$ \\
    \bottomrule
    \end{tabular}
    \caption{\textbf{Backward-only shunting counterfactual.} Five shunting checkpoints per setting, each averaged over the first train, validation, and test diagnostic batches. ``No-I'' removes all inhibitory conductance only from parent resistance during error transport. The gain ratio is the resulting mean distal path gain relative to the recorded path. This is not a physical forward model or training comparison; all forward variables and local eligibility factors remain fixed. Population 0 is the first of two feedforward dendritic populations.}
    \label{tab:backward_only_3f}
    \end{table}

    A stage-resolved operating-point analysis explains the MNIST result. At proximal compartments---the stage whose conductance can gate distal descendants---mean total conductance is $34.5$, mean input resistance is $0.029$, and learned inhibition supplies only $5.2\%$ of total conductance on average (median $0\%$). The mean of the pointwise parent-gain ratio $g^{\mathrm{tot}}/(g^{\mathrm{tot}}-G^I)$ is only $1.068$; it differs from applying this nonlinear ratio to the mean inhibitory fraction. Inhibition is stronger at terminal distal compartments, where it changes forward voltage and local eligibility but has no lower modeled descendants whose path gain it can gate. In the first noise-resilience population, proximal inhibition is about $21\%$ and full removal increases distal gain by $1.281\times$, yet direction remains stable (Table~\ref{tab:backward_only_3f}). A sensitivity extension shows visible directional changes only under nonphysical over-removal, once gains are several-fold larger and the conductance floor is active. These fixed-checkpoint quantities make the null a measured operating-point boundary rather than an unexplained absence of effect; they do not establish a universal inhibitory scale.

    \paragraph{Fixed-state factor decomposition.}
    We next recomputed the exact-error 3F gradient after replacing one local factor at a time while keeping the same recorded inputs, voltages, soma errors, active synapse masks, couplings, and activation derivatives unless that derivative was the factor under study. Table~\ref{tab:fixed_state_factorial} separates two effects that were conflated in the broader architecture comparison. Removing inhibitory conductance from the local or backward input-resistance factor changes scale more than direction. Replacing the reversal-potential driving force by the implementation's threshold-centered voltage proxy also leaves direction largely intact. The diagnostic is sensitive to larger path-factor interventions: setting local resistance to one gives cosine $0.948$ with a $5.56\times$ norm ratio, and removing parent-compartment transfer derivatives gives cosine $0.385$. The restricted matched-width/scalar-fallback feedback remains the largest mismatch. These are diagnostic hybrids at a fixed forward state, not physical forward models.

    \begin{table}[htbp]
    \centering
    \small
    \begin{tabular}{@{}L{0.50\linewidth}cc@{}}
    \toprule
    \textbf{Fixed-state gradient construction} & \textbf{Branch cosine} & \textbf{Norm ratio} \\
    \midrule
    Exact eligibility and effective transport & $1.000\pm0.000$ & $1.00\pm0.00$ \\
    No inhibitory conductance in backward resistance & $0.991\pm0.001$ & $1.10\pm0.04$ \\
    No inhibitory conductance in local eligibility resistance & $0.984\pm0.003$ & $1.29\pm0.03$ \\
    Threshold-centered synaptic voltage proxy & $0.971\pm0.005$ & $0.93\pm0.03$ \\
    Unit local resistance & $0.948\pm0.009$ & $5.56\pm0.53$ \\
    No parent transfer derivatives in transport & $0.385\pm0.040$ & $0.42\pm0.05$ \\
    Matched-width/scalar-fallback feedback & $0.190\pm0.076$ & $0.38\pm0.05$ \\
    \bottomrule
    \end{tabular}
    \caption{\textbf{Factor-by-factor gradient diagnostic at fixed forward state.} Five shunting 3F checkpoints; values are mean $\pm$ s.e.m.\ over seeds on a held-out batch. Branch cosine is the parameter-count-weighted mean of non-somatic block cosines. Each row changes only the named factor relative to exact reconstruction. The input-resistance interventions retain recorded voltages and therefore do not represent a new self-consistent forward pass.}
    \label{tab:fixed_state_factorial}
    \end{table}

    \begin{figure}[htbp]
    \centering
    \includegraphics[width=\textwidth]{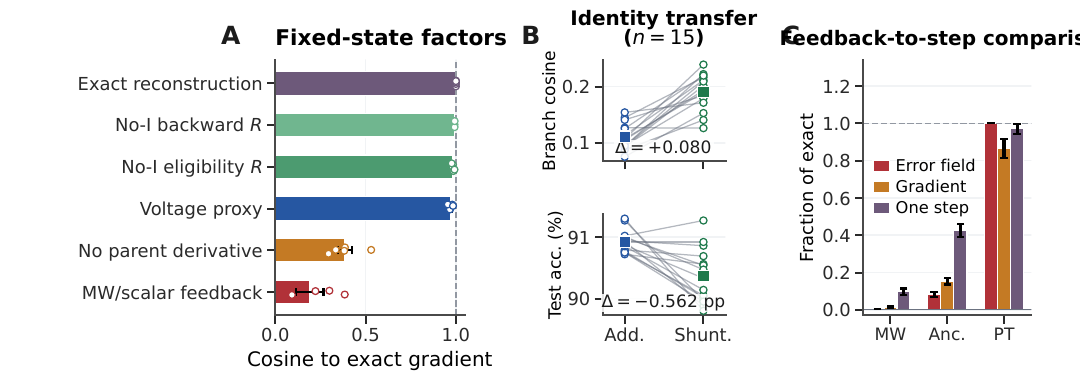}
    \caption{\textbf{Mechanistic analysis separating path effects from learning consequences.}
    \textbf{(A)}~Fixed-forward-state factor substitutions in five shunting checkpoints. Removing inhibitory conductance from either resistance factor leaves gradient direction close to exact reconstruction, whereas omitting parent transfer derivatives or using matched-width/scalar-fallback feedback causes larger directional errors. Bars are mean $\pm$ s.e.m.; points are seeds.
    \textbf{(B)}~Fifteen paired identity-transfer checkpoints. Shunting improves branch-gradient direction, while additive test accuracy remains slightly higher; squares denote means.
    \textbf{(C)}~The same-batch feedback ladder measured as optimally rescaled dendritic error-field capture, eligibility-weighted branch-gradient capture, and norm-matched one-step progress at relative step $10^{-5}$. Bars pool five checkpoints per core after averaging three diagnostic splits; error bars are s.e.m.\ over checkpoints. In this cohort the global scalar equals the matched-width/scalar-fallback field, and the available ancestry field equals exact-soma ancestry. Error geometry alone is therefore not treated as a learning metric. Fixed-state substitutions are diagnostic hybrids, not separately trained forward models.}
    \label{fig:mechanistic_audit}
    \end{figure}

    \paragraph{Norm-matched finite updates.}
    Gradient cosine is invariant to global rescaling, but actual loss change is not. A matched feedback ladder connects the descriptive field metric to its learning consequence (Fig.~\ref{fig:mechanistic_audit}C): ancestry sharing improves raw dendritic error capture, the eligibility-weighted gradient, and norm-matched one-step progress relative to scalar or matched-width/scalar-fallback feedback, while exact transport is strongest. Raw error capture is not numerically interchangeable with update quality because sample-specific eligibility reweights and aggregates the field into parameter gradients.

    We also globally rescaled the matched-width/scalar-fallback 3F branch gradient to the exact branch-gradient norm and applied one update to non-somatic parameters only. The update norm was set to a fixed fraction of the branch-parameter norm. Across the practical step range in Table~\ref{tab:norm_matched_one_step}, both architectures decrease the same held-out batch loss in every seed. The fraction of exact-gradient progress depends on step size, and no tested step gives a reliable shunting--additive difference. This confirms that update scale does not rescue the proposed general directional advantage.

    \begin{table}[htbp]
    \centering
    \small
    \begin{tabular}{@{}lcc@{}}
    \toprule
    \textbf{Step / branch-parameter norm} & \textbf{Shunting} & \textbf{Additive} \\
    \midrule
    $3{\times}10^{-5}$ & $0.225\pm0.102$ (5/5) & $0.207\pm0.030$ (5/5) \\
    $1{\times}10^{-4}$ & $0.269\pm0.089$ (5/5) & $0.344\pm0.029$ (5/5) \\
    $1{\times}10^{-3}$ & $0.384\pm0.065$ (5/5) & $0.603\pm0.163$ (5/5) \\
    \bottomrule
    \end{tabular}
    \caption{\textbf{One-step loss decrease after global branch-norm matching.} Entries are the loss decrease produced by the norm-matched matched-width/scalar-fallback 3F update divided by that produced by the exact-gradient update at the same step norm, reported as mean $\pm$ s.e.m.; parentheses give the number of seeds with positive loss decrease. Five paired checkpoints, one held-out batch per checkpoint. No architecture ordering is consistent across the tested steps, so small-sample inferential tests are omitted. Smaller steps include one shunting seed at numerical/non-smooth resolution; all exact-gradient updates decrease loss.}
    \label{tab:norm_matched_one_step}
    \end{table}

    \begin{figure}[htbp]
    \centering
    \includegraphics[width=\textwidth]{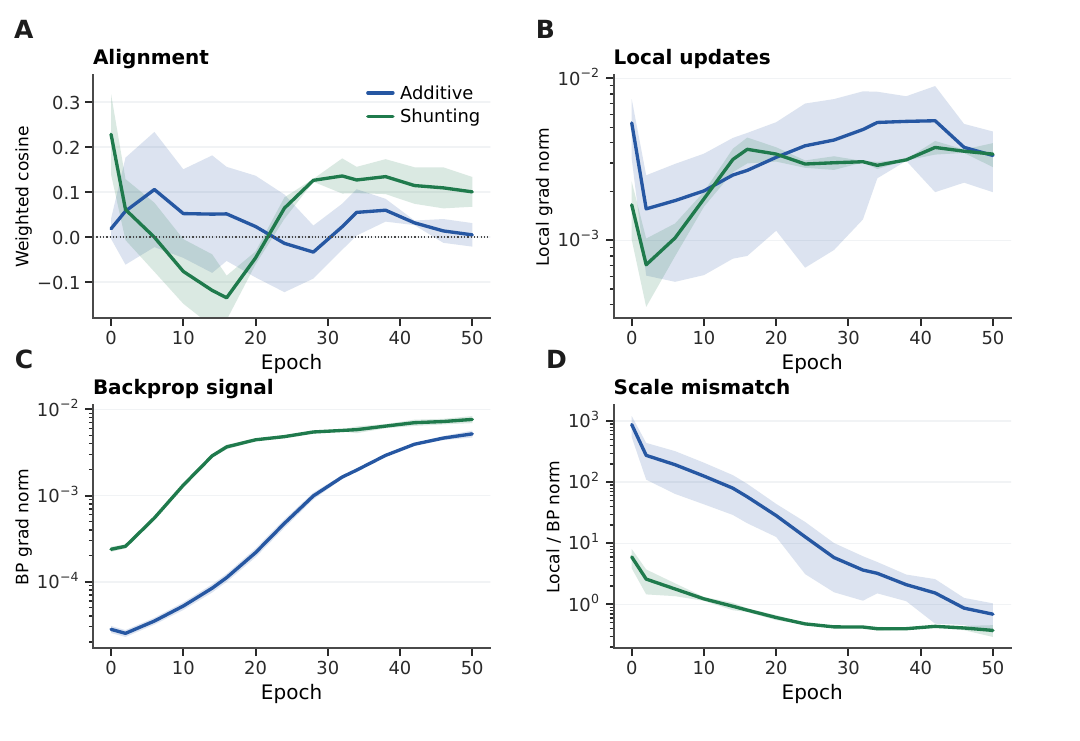}
    \caption{\textbf{Alignment dynamics with explicit gradient-norm checks.}
    \textbf{(A)} Weighted cosine similarity between LocalCA and backpropagation gradients over training for the 5F rule. \textbf{(B,C)} Weighted local and backprop gradient norms stay nonzero for both additive and shunting models. \textbf{(D)} The additive control begins with a much larger scale mismatch, then approaches the backprop scale without recovering strong directional alignment. These are three archived checkpoints and characterize that operating point rather than a general cross-core advantage. Shading is s.d.\ across 3 seeds.}
    \label{fig:alignment_norm_dynamics}
    \end{figure}
    \FloatBarrier
    
    \subsection{Verification and Seed Robustness}
    These checks (Fig.~\ref{fig:verification_appendix}) compare the main seed set with held-out seeds and report the HSIC-weight sensitivity for figure-ground MNIST. The additional matched experiments use the same objective-weight comparison on both seed sets, so seed robustness and HSIC contribution can be summarized as a matched two-factor control.
    
    \begin{figure}[htbp]
    \centering
    \includegraphics[width=\textwidth]{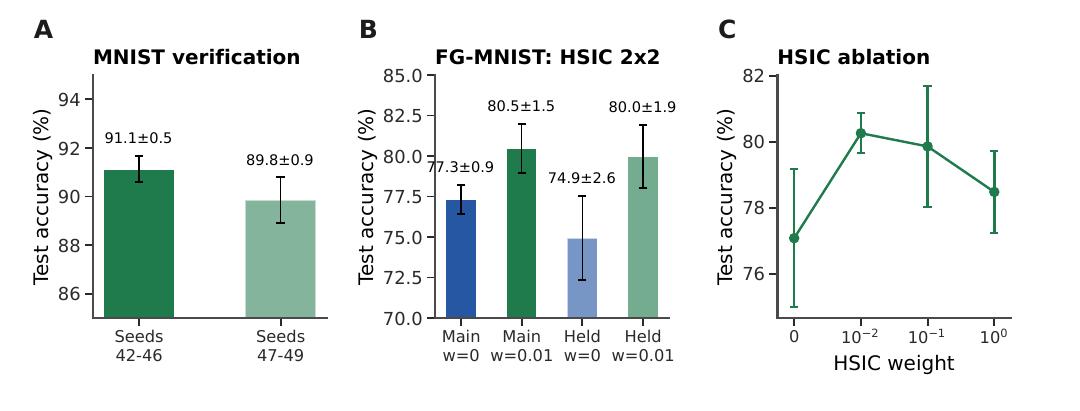}
    \caption{\textbf{Verification and seed robustness (supplementary).} \textbf{(A)}~MNIST verification: repeat runs on main seeds (42--46) yield $91.1\pm0.5$\%; held-out seeds (47--49) yield $89.8\pm0.9$\%, showing robustness across random seeds. \textbf{(B)}~Matched figure-ground MNIST HSIC control: both main seeds (42--46) and held-out seeds (47--49) compare HSIC weight $0$ against $0.01$. \textbf{(C)}~HSIC weight ablation on figure-ground MNIST: moderate weights ($0.01$--$0.1$) perform best.}
    \label{fig:verification_appendix}
    \end{figure}
    \FloatBarrier
    
    \subsection{Additional Stress Tests}
    These stress tests probe two ways in which low-bandwidth broadcast can become unreliable: deeper dendritic paths and noisy teaching signals. They are not separate main benchmarks; they delimit where shunting remains useful and where additive local learning degrades.
    
    \begin{figure}[htbp]
    \centering
    \includegraphics[width=\textwidth]{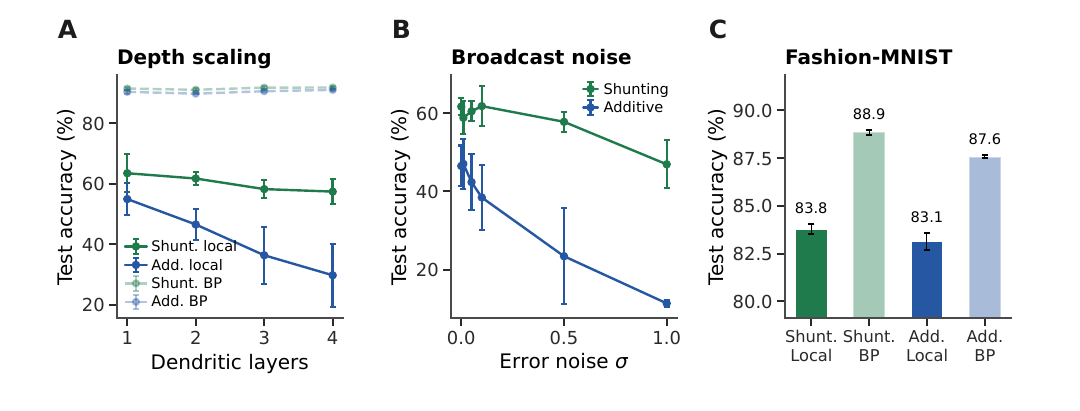}
    \caption{\textbf{Additional stress tests.} \textbf{(A)}~Depth scaling over $1$--$4$ dendritic layers: shunting local (green) degrades more gracefully from 63.5\% to 57.4\%; additive local (blue) drops from 54.9\% to 29.7\%. Matched backpropagation references (dashed) remain stable. \textbf{(B)}~Broadcast-noise robustness, where $\sigma$ is the standard deviation of Gaussian noise added to the broadcast error: shunting remains stable for $\sigma\!\leq\!0.1$; additive degrades rapidly and reaches chance at $\sigma\!=\!1$. \textbf{(C)}~Fashion-MNIST: the shunting gain is modest on this cleaner benchmark (83.8\% local vs.\ 88.9\% backprop), consistent with the regime-dependent interpretation in the main text.}
    \label{fig:additional_stress_tests}
    \end{figure}
    \FloatBarrier
    
    \subsection{FA/DFA Baseline Comparison}
    Figure~\ref{fig:fa_dfa_appendix} includes feedback alignment as an adjacent local-learning reference point. DFA can be evaluated on the dendritic cores, whereas the FA condition is marked unavailable for the block-structured dendritic updates because the required random feedback matrices are not dimensionally compatible with the branch-level parameterization.
    
    \begin{figure}[htbp]
    \centering
    \includegraphics[width=\textwidth]{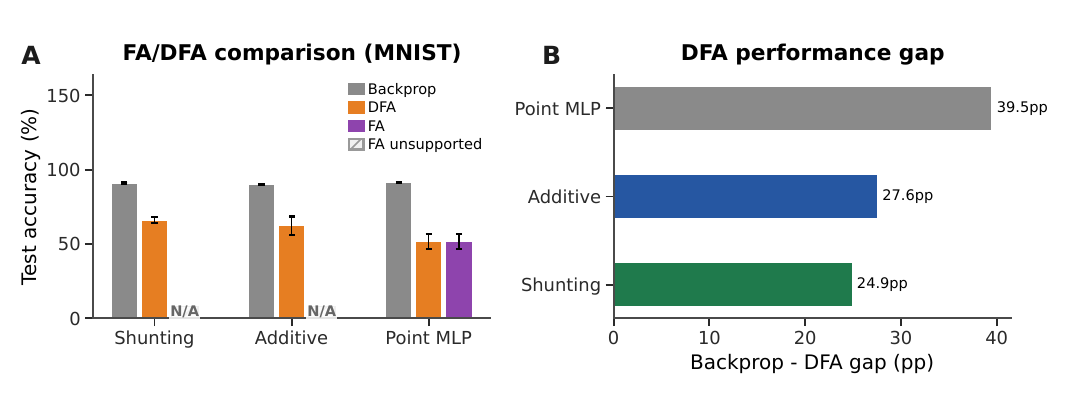}
    \caption{\textbf{Feedback alignment baselines (MNIST).} \textbf{(A)}~Grouped comparison: standard backprop (neutral gray), DFA (orange), and FA (purple). Hatched N/A markers indicate that FA is not implemented for the block-structured dendritic architectures because the random feedback matrices are not dimensionally compatible with the tree-structured branch updates. DFA achieves 66.1\% on shunting vs.\ 62.2\% additive vs.\ 51.8\% point MLP. \textbf{(B)}~Backprop$-$DFA gap: dendritic architectures (24.9--27.6~pp) show smaller gaps than point MLPs (39.5~pp), suggesting conductance-based architecture is partially compatible with random feedback.}
    \label{fig:fa_dfa_appendix}
    \end{figure}
    \FloatBarrier
    
    \subsection{CIFAR-10 Results}
    Figure~\ref{fig:cifar10_mechanism_extension} and Table~\ref{tab:cifar10_mechanism_extension} report a CIFAR-10 control ladder, used as a harder-dataset mechanistic stress test rather than a competitive vision benchmark.
    
    \begin{figure}[htbp]
    \centering
    \includegraphics[width=\textwidth]{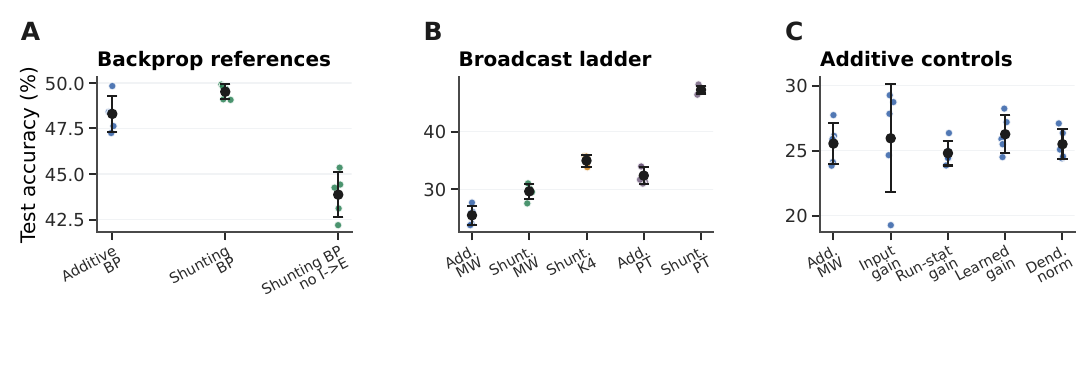}
    \caption{\textbf{CIFAR-10 control ladder in the compact direct-I-stream family.}
    \textbf{(A)}~Matched backpropagation references. With nonnegative inputs and a nonlinear decoder, shunting slightly exceeds additive under standard training, and removing learned I-to-E conductance lowers the shunting reference.
    \textbf{(B)}~Broadcast ladder. Matched-width/scalar-fallback feedback remains weak on this harder dataset, random low-rank $K{=}4$ improves partially, and exact effective transported error nearly reaches the matched shunting backpropagation reference.
    \textbf{(C)}~Additive fairness controls. Input-dependent gain, running-stat gain, learned gain, and dendritic normalization do not rescue the additive matched-width/scalar-fallback baseline. This figure is a harder-dataset stress test, not a claim of competitive CIFAR-10 benchmarking.}
    \label{fig:cifar10_mechanism_extension}
    \end{figure}
    \FloatBarrier
    
    \begin{table}[htbp]
    \centering
    \scriptsize
    \renewcommand{\arraystretch}{1.06}
    \setlength{\tabcolsep}{3pt}
    \begin{tabular}{@{}L{0.32\linewidth}L{0.12\linewidth}L{0.24\linewidth}L{0.10\linewidth}L{0.12\linewidth}@{}}
    \toprule
    \textbf{Condition} & \textbf{Core} & \textbf{Training / broadcast} & \textbf{I-to-E} & \textbf{Test acc.} \\
    \midrule
    standard & additive & backprop & yes & $48.3 \pm 1.0$ \\
    standard & shunting & backprop & yes & $\mathbf{49.5 \pm 0.4}$ \\
    standard, no learned inhibition & shunting & backprop & no & $43.9 \pm 1.2$ \\
    \midrule
    MW/scalar fallback & additive & 5F restricted & yes & $25.5 \pm 1.6$ \\
    MW, input-dependent gain & additive & 5F restricted & yes & $25.9 \pm 4.1$ \\
    MW, running-stat gain & additive & 5F restricted & yes & $24.8 \pm 0.9$ \\
    MW, learned gain & additive & 5F restricted & yes & $26.3 \pm 1.5$ \\
    MW, dendritic normalization & additive & 5F restricted & yes & $25.5 \pm 1.2$ \\
    MW/scalar fallback & shunting & 5F restricted & yes & $29.7 \pm 1.3$ \\
    MW, no learned inhibition & shunting & 5F restricted & no & $22.3 \pm 0.7$ \\
    \midrule
    transported error & additive & 5F path transport & yes & $32.4 \pm 1.5$ \\
    rank-$4$ & shunting & 5F low-rank & yes & $35.0 \pm 1.0$ \\
    transported error & shunting & 5F path transport & yes & $\mathbf{47.2 \pm 0.6}$ \\
    transported error, no learned inhibition & shunting & 5F path transport & no & $38.6 \pm 1.0$ \\
    \bottomrule
    \end{tabular}
    \renewcommand{\arraystretch}{1.0}
    \caption{\textbf{Exact CIFAR-10 values for the 70-run control ladder (5 seeds per condition, validation-best epoch).} All rows use flattened CIFAR-10 in $[0,1]$, nonnegative transfer outputs, positive conductance transforms, a compact depth-$4$ dendritic core, and a nonlinear decoder. LocalCA maps output error back to decoder-input/soma coordinates through the decoder Jacobian. The useful conclusions are narrow: learned shunting slightly improves the matched backpropagation reference, learned I-to-E conductance is necessary for the strongest shunting results, transported error nearly closes the shunting LocalCA gap, and additive matched-width/scalar-fallback controls do not explain the transported-shunting result.}
    \label{tab:cifar10_mechanism_extension}
    \end{table}
    
    \subsection{Additive Normalization Control}
    This control separates generic additive voltage normalization from shunting conductance dynamics in a lower-capacity MNIST stress regime.
    
    \begin{figure}[htbp]
    \centering
    \includegraphics[width=\textwidth]{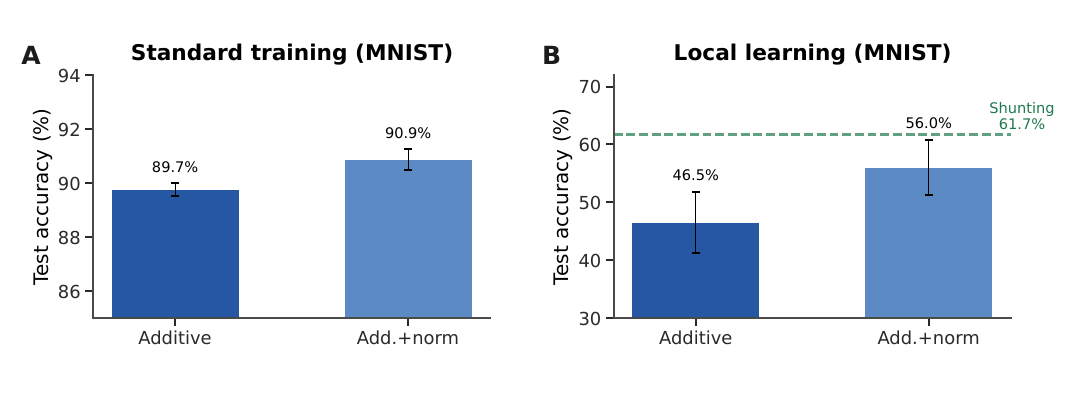}
    \caption{\textbf{Additive + normalization control (MNIST stress-regime diagnostic).} \textbf{(A)}~Standard backprop: normalization provides a small boost (89.7\% $\to$ 90.9\%). \textbf{(B)}~Local learning: normalization improves additive from 46.5\% to 56.0\% ($+9.5$~pp), partially closing the gap to shunting (61.7\%, dashed green). These lower local accuracies come from a separate stress-regime diagnostic, not from the matched-capacity performance setting in Table~\ref{tab:gap_closing}.}
    \label{fig:additive_norm_control}
    \end{figure}
    
    \subsection{Morphology $\times$ Inhibition Regime Map}
    
    This regime map (Fig.~\ref{fig:morphology_ie_regime}) describes how performance varies jointly with tree geometry and inhibitory input. It does not isolate a causal morphology effect: under the analytical initialization used in this sweep, one dendritic population in inhibited shunting models begins with a coupling stage at the $10^{-6}$ floor, whereas additive stages remain active. The figure therefore reports the performance interaction without using it as mechanistic evidence that morphology alone selects the useful inhibitory regime.
    
    \begin{figure}[htbp]
    \centering
    \includegraphics[width=\textwidth]{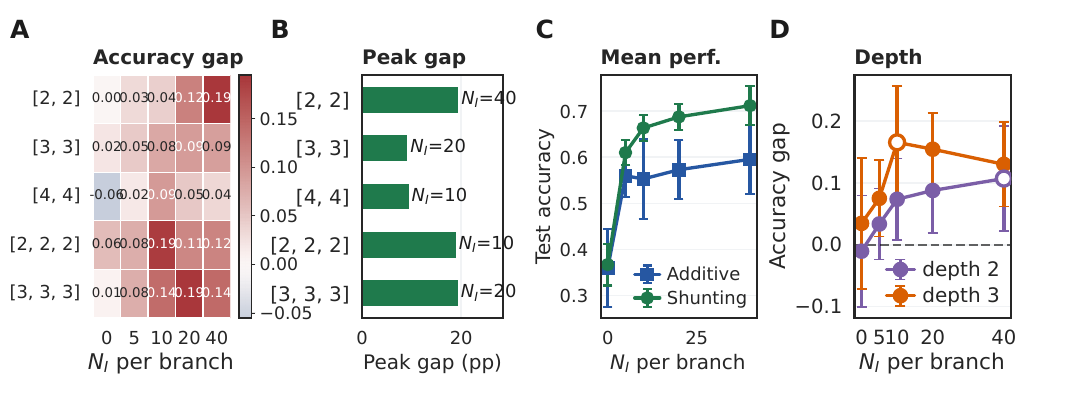}
    \caption{\textbf{Descriptive morphology $\times$ inhibition performance map.}
    \textbf{(A)}~Shunting-minus-additive LocalCA accuracy gap across dendritic morphologies and inhibitory synapse counts.
    \textbf{(B)}~Best inhibitory count for each morphology.
    \textbf{(C)}~Mean performance across tested trees.
    \textbf{(D)}~Average shunting advantage by depth, with markers showing the peak inhibitory count for each depth. The grid uses nonnegative inputs. Because analytical initialization places one shunting population at a near-zero coupling floor, these comparisons are descriptive rather than a causal morphology intervention.}
    \label{fig:morphology_ie_regime}
    \end{figure}
    
    \subsection{Input-Mode Probe: Direct Inhibitory Stream vs.\ Explicit Inhibitory Cells}
    
    The main experiments use input-driven inhibitory streams: the transfer layer provides nonnegative excitatory and inhibitory channels from the external input, and excitatory dendrites receive learned I-to-E conductances directly. To test whether the same mechanism can be generated by an explicit feedforward inhibitory population, we ran a one-feedforward-layer noise-resilience probe that preserves the relevant dendritic depth ($[3,3]$ morphology). In this setting, explicit inhibitory cells reach $94.2\pm0.2\%$ under path-transport LocalCA when their dendrites use the same update policy as excitatory dendrites, $93.8\pm0.2\%$ when those inhibitory-cell updates are frozen, and $95.1\pm0.3\%$ under standard backprop. Direct input-driven inhibition gives the expected feedback-fidelity ladder on the same one-layer architecture: matched-width/scalar-fallback feedback reaches $85.2\pm0.7\%$, while exact effective path transport reaches $95.2\pm0.0\%$. Thus explicit inhibitory cells can generate useful path-gain structure once the architecture isolates the dendritic-path question.
    
    \begin{table}[!htbp]
    \centering
    \small
    \renewcommand{\arraystretch}{1.02}
    \setlength{\tabcolsep}{5pt}
    \begin{tabular}{@{}L{0.32\linewidth}L{0.39\linewidth}L{0.18\linewidth}@{}}
    \toprule
    \textbf{Inhibitory input mode} & \textbf{Training / broadcast} & \textbf{Test accuracy} \\
    \midrule
    Direct inhibitory stream & LocalCA, MW/scalar-fallback feedback & $0.852 \pm 0.007$ \\
    Direct inhibitory stream & LocalCA, exact effective path transport & $\mathbf{0.952 \pm 0.000}$ \\
    Explicit I cells & LocalCA, path transport, train I-cell dendrites & $0.942 \pm 0.002$ \\
    Explicit I cells & LocalCA, path transport, freeze I-cell dendrites & $0.938 \pm 0.002$ \\
    Explicit I cells & standard backprop & $0.951 \pm 0.003$ \\
    \bottomrule
    \end{tabular}
    \caption{\textbf{One-feedforward-layer input-mode path-gain probe on noise resilience (3 seeds).} All rows use one excitatory dendritic population with $[3,3]$ morphology and nonnegative inputs. The direct-stream rows have no explicit inhibitory neurons; the input stream drives learned branch-level inhibitory conductances directly. The explicit-I rows use a matched inhibitory population with $[3,3]$ morphology. Exact effective path transport closes the direct-stream restricted-feedback gap, and explicit inhibitory cells nearly match both direct-stream path transport and the explicit-I matched backpropagation reference.}
    \label{tab:input_mode_probe}
    \end{table}
    
    \paragraph{Post-training inhibition intervention and exact-error geometry.}
    We next tested whether learned inhibitory conductance carries sample-dependent structure rather than only changing global scale (Fig.~\ref{fig:inhibition_causality_error_rank}). We replayed trained shunting LocalCA checkpoints while replacing the branch-level inhibitory conductance with four post-training interventions: zero inhibition, sample-shuffled inhibition, per-branch batch means, or one uniform matched mean. On a $512$-example held-out subset, learned MNIST inhibition at $N_I{=}5$ gives $90.4\pm1.1\%$ accuracy, while the four interventions drop to $19.8$--$26.4\%$. On noise resilience at $N_I{=}10$, learned inhibition gives $80.5\pm2.4\%$, while the interventions drop to $10.3$--$19.3\%$. Thus the inhibitory field is not interchangeable with a matched global conductance load. We retain the unrestricted SVD summaries as descriptive properties of the complete error matrices, not as tests of the matched-width/scalar-fallback feedback. A stage-resolved analysis showed that the pooled feedback cosine was dominated by the width-matched somatic stage, whose cosine is one by construction and whose error-energy share differs across cores. On distal and proximal compartments the matched-width/scalar-fallback field has low cosine in both architectures; the pooled comparison is therefore not evidence for a shunting-specific compatibility advantage.
    
    We also computed the direct log-gain covariance margin in Eq.~\eqref{eq:compression_condition} on the selected morphology diagnostic checkpoints. The margin is zero when no inhibitory conductance is present and slightly negative in the inhibited shunting selections ($-2.5\times10^{-3}$, $-5.0\times10^{-4}$, and $-2.8\times10^{-4}$ at $N_I{=}5,20,40$). We therefore do not use this covariance condition as a positive main result. Its role is to clarify when shunting should compress gains. The measured path-gain dispersion and unrestricted exact-error rank remain descriptive, while the intervention establishes only that the learned inhibitory field is sample-dependent and necessary to the trained forward computation; none of these results establishes a general feedback-compatibility advantage.
    
    \begin{figure}[htbp]
    \centering
    \includegraphics[width=\textwidth]{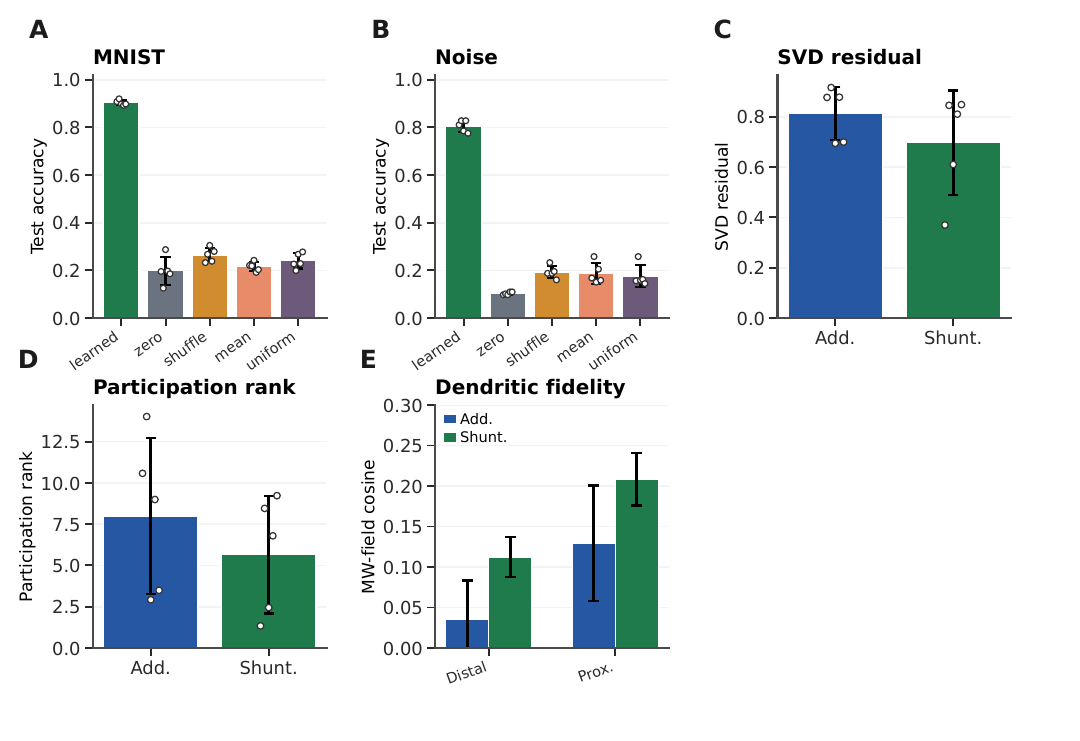}
    \caption{\textbf{Post-training inhibition intervention and exact-error geometry.}
    \textbf{(A,B)} Post-training inhibitory-conductance interventions on trained shunting LocalCA checkpoints. Accuracy is measured on a $512$-example held-out subset; dots show seeds. Removing, shuffling, or clamping learned inhibition collapses performance, showing that the learned inhibitory field is sample-dependent and functionally important. \textbf{(C,D)} Unrestricted exact-error geometry on three matched MNIST $N_I{=}5$ checkpoints per core; these SVD summaries do not test the implemented feedback. \textbf{(E)} Stage-resolved matched-width/scalar-fallback field cosine on distal and proximal compartments, excluding the exact-by-construction width-matched somatic stage. Bars average train, validation, and test batches within each of five checkpoints; error bars are s.d.\ over checkpoints.}
    \label{fig:inhibition_causality_error_rank}
    \end{figure}
    
    \subsection{Experimental Predictions}
    The model is not a detailed cell-type circuit model, but its path-gain identity gives falsifiable dendritic predictions (Table~\ref{tab:experimental_predictions}). These predictions are naturally aligned with branch-level dendritic computation and compartment-targeted inhibition in pyramidal neurons \cite{spruston2008pyramidal,branco2010single,larkum2013cellular,hennequin2017inhibitory}.
    
    \begin{table}[t]
    \centering
    \small
    \begin{tabular}{@{}L{0.31\linewidth}L{0.58\linewidth}@{}}
    \toprule
    \textbf{Model implication} & \textbf{Biological prediction} \\
    \midrule
    Path-local inhibitory conductance changes $\alpha_n^{\mathrm{cond}}$ only for descendants whose credit path crosses the inhibited compartment. & Branch-local inhibitory perturbation should distort plasticity or teaching-signal efficacy for synapses below the perturbed compartment more than for sister branches. \\
    The log derivative $\partial\log\alpha_n^{\mathrm{cond}}/\partial G_k^I=-R_k^{\mathrm{tot}}$ predicts stronger credit gating in high-resistance compartments. & The same inhibitory conductance change should have larger effects on low-conductance/high-resistance branches than on already high-conductance branches. \\
    Single-channel broadcast succeeds when exact compartment errors are effectively shared across the relevant branches and fails when branch identity matters. & Tasks requiring route- or context-specific branch credit should need richer dendritic teaching signals than a single scalar modulatory broadcast. \\
    \bottomrule
    \end{tabular}
    \caption{\textbf{Experimental predictions from the path-gain mechanism.} These are qualitative biological predictions, not claims tested directly in this paper.}
    \label{tab:experimental_predictions}
    \end{table}
    
    \subsection{Cue-Routing Feedback-Rank Diagnostic}
    
    Cue routing tests a regime where one shared teaching signal is expected to fail: context determines which noisy cue is reliable, so different pathways should receive different credit. Figure~\ref{fig:cue_routing} shows that the benchmark is learnable and that moving beyond the matched-width/scalar-fallback baseline helps, but the path-structured variant does not yet beat the best unstructured low-rank setting. We therefore interpret this result as a failure-mode and rank-requirement diagnostic, not as a solved structured-feedback method.
    
    \begin{figure}[!p]
    \centering
    \includegraphics[width=\textwidth]{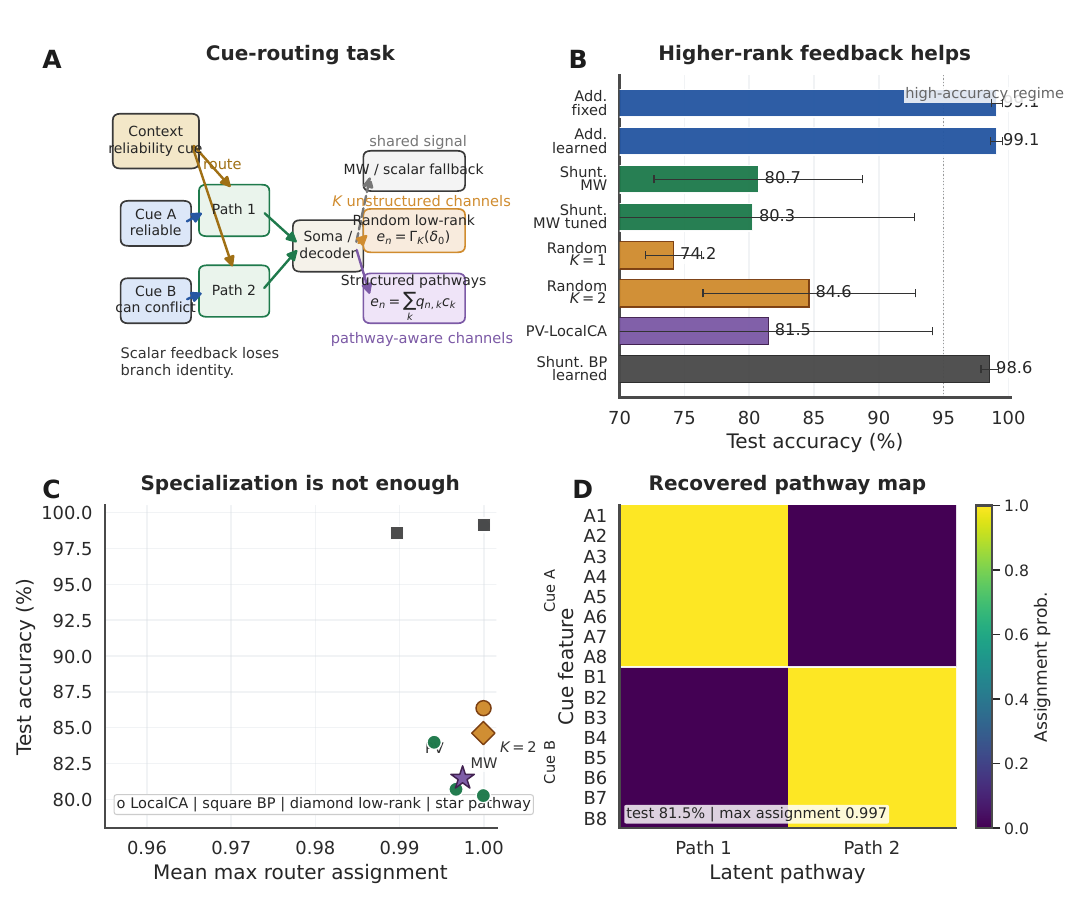}
    \caption{\textbf{Structured feedback when restricted feedback fails.}
    \textbf{(A)}~Cue-routing task schematic. Context determines which noisy cue is reliable on each trial, so different dendritic pathways should receive different teaching signals.
    \textbf{(B)}~Five-seed comparison of matched-width/scalar-fallback, random low-rank, and pathway-structured feedback. Moving beyond the restricted baseline helps, while the best structured repair remains open.
    \textbf{(C)}~Test accuracy vs.\ router specialization across learned-router conditions. High specialization alone does not guarantee success.
    \textbf{(D)}~Recovered pathway assignments for a representative learned-router run. Cue-A and cue-B features segregate into separate latent pathways even when feedback quality, rather than router discreteness, is the limiting factor.}
    \label{fig:cue_routing}
    \end{figure}
    \FloatBarrier
    
    \subsection{Random Low-Rank Broadcast Channels}
    These controls test whether failures of the matched-width/scalar-fallback baseline reflect insufficient feedback bandwidth rather than a failure of local eligibility itself. Table~\ref{tab:low_rank_broadcast} gives the exact values behind the noise-resilience ladder in Fig.~\ref{fig:rule_feedback_controls}D and the cue-routing rank sweep in Fig.~\ref{fig:cue_routing}B.
    
    \begin{table}[!htbp]
    \centering
    \small
    \renewcommand{\arraystretch}{1.02}
    \setlength{\tabcolsep}{5pt}
    \begin{tabular}{@{}lll@{}}
    \toprule
    \textbf{Setting} & \textbf{Broadcast} & \textbf{Test accuracy} \\
    \midrule
    \multicolumn{3}{@{}l}{\emph{Noise resilience, shunting, rank bridge}} \\
    MW/scalar fallback & main restricted feedback & $0.662 \pm 0.016$ \\
    path propagation & recursive attenuation proxy & $0.762 \pm 0.047$ \\
    low-rank & $K{=}1$ & $0.453 \pm 0.126$ \\
    low-rank & $K{=}2$ & $0.630 \pm 0.042$ \\
    low-rank & $K{=}4$ & $0.764 \pm 0.025$ \\
    low-rank & $K{=}8$ & $0.795 \pm 0.019$ \\
    path transport & exact effective path transport & $\mathbf{0.847 \pm 0.018}$ \\
    \addlinespace[2pt]
    \multicolumn{3}{@{}l}{\emph{Cue routing, learned shunting router, rank/structure sweep}} \\
    MW/scalar fallback & tuned restricted feedback & $0.803 \pm 0.125$ \\
    low-rank & $K{=}1$ & $0.742 \pm 0.022$ \\
    low-rank & $K{=}2$ & $\mathbf{0.846 \pm 0.082}$ \\
    low-rank & $K{=}4$ & $0.748 \pm 0.085$ \\
    pathway-vector & tuned structured rank-$2$ & $0.815 \pm 0.126$ \\
    \bottomrule
    \end{tabular}
    \renewcommand{\arraystretch}{1.0}
    \caption{\textbf{Rank bridge and rank-structure comparison.} On shunting noise resilience, both recursive path propagation and higher-rank random broadcast close a substantial fraction of the gap between matched-width/scalar-fallback feedback and exact effective path transport, with low-rank continuing to improve up to $K{=}8$. On cue routing, the learned-router sweep shows that moving beyond the restricted baseline helps, but the best unstructured low-rank setting ($K{=}2$) is still slightly better than the tuned structured pathway-vector variant. We therefore interpret routed credit assignment as an open higher-rank problem rather than as a solved pathway-vector rescue.}
    \label{tab:low_rank_broadcast}
    \end{table}

    \section{Dendritic Network Implementation}
    \label{app:implementation}
    This section gives the implementation-level specification needed to reproduce the dendritic architectures and LocalCA updates. It separates the biological modeling assumptions from software choices such as positive raw-parameter transforms, decoder-error coordinates, and optimizer gradient assignment.
    
    \subsection{Implementation Overview}
    
    Each feedforward dendritic layer is a population of dendritic neurons. A neuron contains a rooted tree of compartments; each compartment receives excitatory inputs, optionally inhibitory inputs, and dendritic coupling conductances from its children. The main experiments use branch-level input-driven inhibition: a nonnegative transfer stage produces excitatory and inhibitory input streams from the external input, and the inhibitory stream drives learned inhibitory conductances on excitatory dendritic branches. The explicit-I probe in Table~\ref{tab:input_mode_probe} instead instantiates a separate feedforward inhibitory population, whose output projects inhibitory conductance onto excitatory dendrites.
    
    The additive and shunting cores are architecture-matched: they use the same tree, synapse counts, nonnegative transfer stage, positive conductance parameterization, optimizer protocol, and seed sets. They differ only in how a branch combines excitatory, inhibitory, and dendritic inputs. Decoder weights are unconstrained readout parameters and are not interpreted as conductances. Table~\ref{tab:eq_code_map} maps each implementation concept to its governing equation, code location, and update scope.
    
    \begin{table}[htbp]
    \centering
    \scriptsize
    \setlength{\tabcolsep}{3pt}
    \renewcommand{\arraystretch}{1.08}
    \begin{tabular}{@{}L{0.21\linewidth}L{0.13\linewidth}L{0.34\linewidth}L{0.26\linewidth}@{}}
    \toprule
    \textbf{Concept / parameter family} & \textbf{Paper eq.} & \textbf{Code location} & \textbf{Update scope / notes} \\
    \midrule
    Shunting branch voltage & Eq.~\eqref{eq:voltage} & \path{DendriticBranchLayer.forward} (\path{use_shunting=True}) & Numerator/denominator conductance integration; in-theorem. \\
    Additive control voltage & Eq.~\eqref{eq:voltage} & \path{DendriticBranchLayer.forward} (\path{use_shunting=False}); \path{normalize_additive_voltage} & Same tree/synapses; inhibition subtractive rather than a denominator load. \\
    E/I synaptic and dendritic-coupling conductances & Cor.~\ref{cor:factorization}; Eqs.~\eqref{eq:3f}--\eqref{eq:5f} & \path{LocalCreditAssignment._apply_local_rule_gradients} (conductance branch) & LocalCA eligibility $xR^{\mathrm{tot}}(E{-}V)\,e$, then positive-transform derivative; covered by the theorem. \\
    Additive eligibility (matched control) & Eq.~\eqref{eq:3f} & \path{_apply_local_rule_gradients} (additive branch) & Matched signed additive derivative, with optional additive-gain controls. \\
    Broadcast-mode switch & broadcast $e_n$ & \path{LocalRuleConfig.error_broadcast_mode}; \path{_compute_low_rank_broadcast}; \path{_precompute_path_transport_errors}; \path{_precompute_path_propagation_factors} & Selects global scalar, matched-width/scalar-fallback, ancestry-shared, low-rank, path-transport, or pathway-vector feedback. \\
    4F / 5F reliability gates & $r_n^{\mathrm{4F}},\phi_n$ & \path{_compute_layer_rho}, \path{_compute_layer_phi*}, \path{clamp_phi}, \path{FiveFactorConfig} & Slow bounded preconditioners, not additional task-error channels; outside the theorem. \\
    Reactivation transfer & $a_n=f_n(V_n)$ & configured local/autograd schedule & Local reactivation derivative, with $f_n'(V_n)=1$ when disabled; implementation-level training component. \\
    Top-$k$ synaptic input groups & --- & masked \path{_apply_local_rule_gradients} & Sparse input selection; masks multiply the branch eligibility. \\
    Linear / nonlinear decoder & $\delta_0^a{=}W_{\mathrm{dec}}^{\!\top}\delta^y$ / $J_{\mathrm{dec}}^{\!\top}\delta^y$ & local decoder gradient or autograd & Maps output error into soma/core coordinates; implementation-level. \\
    HSIC auxiliary & --- & auxiliary gradient on selected representations & Batch dependence regularizer; auxiliary objective, not a local theorem term. \\
    Router / pathway roles & pathway-vector & configuration-dependent router/pathway updates & Exploratory higher-rank routing; outside the theorem. \\
    Experiment knobs & --- & \path{LocalRuleConfig} YAML (\path{rule_variant}, \path{error_broadcast_mode}, \path{broadcast_rank}, \path{five_factor}) & Reproduces rule family, feedback bandwidth, and 5F stabilizer settings. \\
    \bottomrule
    \end{tabular}
    \caption{\textbf{Implementation map: concepts, equations, code, and update scope.} The exact factorization (Cor.~\ref{cor:factorization}) covers conductance-stage branch and dendritic-coupling parameters; decoder, reactivation, top-$k$, HSIC, and router mechanisms are implementation-level training components documented alongside the theorem rather than derived from it. Code anchors are symbolic names so that minor refactors do not invalidate the map.}
    \label{tab:eq_code_map}
    \end{table}
    
    \paragraph{Implementation-specific update paths.}
    Top-$k$ selector scores are learned selection parameters; the selected masks gate active presynaptic drives and the corresponding local gradients, while inactive weights receive no LocalCA update unless an explicit inactive-update control is enabled. The discrete selection itself is treated as an implementation-level sparsity mechanism rather than as a conductance variable in the theorem. Reactivation slope and bias parameters use the configured optimizer schedule; the conductance update still uses the local derivative $f_n'(V_n)$ for voltage-error conversion. Router/pathway parameters in the cue-routing controls are trained through their configured router/pathway objective and are not included in the conductance-factorization theorem. HSIC gradients are applied to selected representation tensors through the configured local-learning representation pathway; they do not replace the LocalCA conductance eligibility and are reported as an auxiliary objective.
    
    \subsection{Morphology and Synaptic Inputs}
    
    The morphology notation $[b_1,b_2,\dots,b_D]$ specifies a rooted dendritic tree. The factor $b_d$ is the fan-in from dendritic stage $d$ to the next more proximal stage. Thus $[3,3]$ gives $3$ proximal branches per soma and $3$ distal branches per proximal branch, for $9$ distal leaves per soma. In the main feedforward models, synaptic inputs terminate on dendritic branches rather than directly on the soma. Unless a control explicitly changes this, each branch has $N_E$ excitatory synapses and $N_I$ inhibitory synapses. A learned top-$k$ selector makes each branch input sparse; inactive synapses are masked out in the local update unless the corresponding control enables inactive-weight updates.
    
    For a branch $n$, write $x_{nj}^E,x_{nj}^I\geq 0$ for the selected excitatory and inhibitory presynaptic drives and $a_c=f_c(V_c)$ for the reactivated activity of child compartment $c$. The branch-level excitatory, inhibitory, and dendritic conductance currents are
    \[
    U_n^E=\sum_{j=1}^{N_E} g_{nj}^E x_{nj}^E,
    \qquad
    U_n^I=\sum_{j=1}^{N_I} g_{nj}^I x_{nj}^I,
    \qquad
    D_n=\sum_{c\in\mathrm{child}(n)} g_{c\to n}^{\mathrm{den}} a_c .
    \]
    The corresponding dendritic coupling load is
    \[
    G_n^{\mathrm{den}}=\sum_{c\in\mathrm{child}(n)} g_{c\to n}^{\mathrm{den}}.
    \]
    For leaf branches, $D_n=G_n^{\mathrm{den}}=0$. For direct input-driven inhibition, $x^I$ is produced by the inhibitory transfer stream. For explicit-I controls, $x^I$ is the activity of the inhibitory-cell population.
    
    \subsection{Shunting and Additive Forward Equations}
    
    In the shunting core, excitatory reversal is $1$, leak and inhibitory reversal are $0$, and leak conductance is fixed to $1$. The implemented branch voltage is therefore
    \[
    V_n^{\mathrm{shunt}}
    =\frac{U_n^E+D_n}
    {1+U_n^E+U_n^I+G_n^{\mathrm{den}}+\varepsilon}.
    \]
    The small numerical $\varepsilon$ is used only for floating-point safety; analytically the leak conductance keeps the denominator positive. In implementation equations, the same stabilized denominator is used for the recorded input resistance:
    \[
    g_n^{\mathrm{tot}}=1+U_n^E+U_n^I+G_n^{\mathrm{den}},
    \qquad
    g_{n,\varepsilon}^{\mathrm{tot}}=g_n^{\mathrm{tot}}+\varepsilon,
    \qquad
    R_{n,\varepsilon}^{\mathrm{tot}}=(g_{n,\varepsilon}^{\mathrm{tot}})^{-1}.
    \]
    The main analytical equations write $R_n^{\mathrm{tot}}=(g_n^{\mathrm{tot}})^{-1}$; exact-gradient reconstruction tests use the implementation convention above.
    In the additive control, inhibition is subtractive rather than shunting. With the same nonnegative conductance variables and nonnegative inputs, the implemented comparator is
    \[
    V_n^{\mathrm{add}}=U_n^E+D_n-U_n^I,
    \]
    with optional additive-normalization controls reported separately in Fig.~\ref{fig:additive_norm_control}. Thus additive inhibition changes the signed voltage contribution and local additive eligibility, but it does not enter the denominator and does not gate conductance-stage path gain through $R_n^{\mathrm{tot}}$; Table~\ref{tab:shunting_additive_mechanism} summarizes the matched comparison.
    
    \begin{table}[htbp]
    \centering
    \small
    \setlength{\tabcolsep}{4pt}
    \renewcommand{\arraystretch}{1.02}
    \begin{tabular}{@{}L{0.20\linewidth}L{0.36\linewidth}L{0.36\linewidth}@{}}
    \toprule
    \textbf{Property} & \textbf{Shunting / conductance core} & \textbf{Additive control} \\
    \midrule
    Voltage & Normalized by total conductance & Signed E/I voltage sum \\
    Inhibition & Divisive conductance load & Subtractive signed term \\
    Synaptic eligibility & $x R^{\mathrm{tot}}(E-V)e$ & $sxe$, $s\in\{+1,-1\}$ \\
    Driving force & Present through $E-V$ & Absent from additive derivative \\
    Input resistance & State-dependent, $R^{\mathrm{tot}}=(g^{\mathrm{tot}})^{-1}$ & Effectively fixed \\
    Credit interpretation & Inhibition changes path gain and voltage sensitivity & Inhibition shifts voltage but does not gate path gain \\
    \bottomrule
    \end{tabular}
    \caption{\textbf{Matched shunting/additive comparison.} Both architectures use their own local derivative, so the comparison changes the forward-pass integration rule rather than applying a shunting-derived update to an additive model.}
    \label{tab:shunting_additive_mechanism}
    \end{table}
    
    \subsection{Positive Parameterization and Raw-Parameter Gradients}
    
    Trainable synaptic and dendritic conductances are stored as unconstrained raw parameters $\theta$ and transformed by a positive map $g=\psi(\theta)$, with \texttt{softplus} used in the main reported runs:
    \[
    \psi(\theta)=\log(1+\exp\theta),
    \qquad
    \psi'(\theta)=\frac{1}{1+\exp(-\theta)}.
    \]
    Equations in the main text give conductance-space gradients $\partial L/\partial g$. Before assigning gradients to the optimizer variables, the implementation applies the chain rule
    \[
    \frac{\partial L}{\partial \theta}
    =\frac{\partial L}{\partial g}\,\psi'(\theta).
    \]
    The exact-gradient reconstruction diagnostic includes this parameterization factor, as well as top-$k$ activity masks and dendritic block-linear parameterization. This is why the reconstruction comparison is against raw autograd gradients rather than only against abstract conductance variables.
    
    \subsection{Reactivation}
    
    The conductance stage produces a pre-reactivation voltage $V_n$, after which each branch transmits $a_n=f_n(V_n)$. When reactivation is disabled, $f_n$ is the identity and $f_n'(V_n)=1$. The main sweeps use the learnable bounded transform
    \[
    f_n(V)=\frac{\tanh(m_n(V-b_n))+1}{2},
    \qquad
    m_n=\exp(\ell_{m,n}),
    \]
    with derivative
    \[
    f_n'(V)=\frac{m_n}{2}\left[1-\tanh^2(m_n(V-b_n))\right].
    \]
    Exact path transport uses the effective gain $\tilde{\alpha}_n$: parent-to-child transport includes the parent reactivation derivative, and the local update converts activation-space error into pre-reactivation voltage error by multiplying by the child derivative. When reactivation is disabled, this reduces to conductance-only transport through $\alpha_n^{\mathrm{cond}}$.
    
    \subsection{Decoder and Error Coordinates}
    
    Let $h$ denote the final soma/core activity passed to the decoder and let $\delta^y=\partial L/\partial \hat{y}$. For a linear decoder $\hat{y}=W_{\mathrm{dec}}h$, the soma/core activation-space teaching error is
    \[
    \delta_0^a=W_{\mathrm{dec}}^\top\delta^y.
    \]
    For a nonlinear decoder, used in the CIFAR-10 compact control ladder, the implementation uses the decoder-input Jacobian product
    \[
    \delta_0^a=J_{\mathrm{dec}}(h)^\top\delta^y.
    \]
    All LocalCA broadcast modes are defined in this soma/core activation-error coordinate system before local conversion to voltage-space errors. Decoder-update modes are: \textbf{backprop}, which uses autograd for decoder parameters; \textbf{local}, which is implemented for the linear readout and assigns
    \[
    \widehat{\nabla}_{W_{\mathrm{dec}}}L=\left\langle \delta^y h^\top\right\rangle_B,
    \qquad
    \widehat{\nabla}_{b_{\mathrm{dec}}}L=\left\langle \delta^y\right\rangle_B;
    \]
    and \textbf{frozen}, which holds the decoder fixed. The local decoder gradient is output-space; $\delta_0^a$ is the mapped soma/core activation-space error used by dendritic broadcasts.
    
    \subsection{Biological Plausibility Assumptions}
    
    The model makes the following explicit assumptions. First, it uses the steady-state solution of the passive cable equation rather than temporal membrane dynamics; all path gains and gradient diagnostics refer to conductance-stage voltages at steady state. Second, each dendritic cell is a rooted tree, so each compartment has a unique path to the soma and the closed-form path gains $\alpha_n^{\mathrm{cond}}$ and $\tilde{\alpha}_n$ are well-defined. Third, synaptic and dendritic conductance parameters are nonnegative by construction. Fourth, the main sweeps enforce nonnegative first-layer presynaptic drive through nonnegative inputs or a ReLU transfer stage. Fifth, reactivation is handled through local $f'(V)$ factors along the path, with $f'(V)=1$ when disabled. Sixth, each dendritic synapse receives a restricted non-local feedback field---global scalar, matched-width/scalar-fallback, ancestry-shared, low-rank, path-structured, or transported-oracle depending on the condition---while presynaptic activity, voltage, reversal potential, input resistance, and slowly estimated branch-level modulators are local or branch-local quantities.
    
    \subsection{Units and Parameterization}
    Table~\ref{tab:units_parameterization} summarizes the numerical conventions used throughout the implementation and diagnostics, including which quantities are constrained conductances and which are unconstrained readout parameters.
    \begin{table}[t]\centering\small
    \begin{tabular}{@{}llL{0.62\textwidth}@{}}\toprule
    Quantity & Symbol & Convention\\\midrule
    Voltage & $V$ & Conductance-stage shunting voltage in $[0,1]$ under nonnegative drive; additive controls may be signed\\
    Conductances & $g^{\mathrm{syn}}, g^{\mathrm{den}}$ & Nonnegative via $g=\mathrm{softplus}(\theta)$; local gradients are mapped to raw $\theta$ by the transform derivative\\
    Leak conductance & $g^{\mathrm{leak}}$ & Set to $1$\\
    Input resistance & $R^{\mathrm{tot}}$ & $\leq 1$\\\bottomrule
    \end{tabular}
    \caption{Units and normalization.}
    \label{tab:units_parameterization}
    \end{table}
    
    \subsection{Representative Manuscript Architectures}
    \begin{table}[htbp]\centering\scriptsize
    \renewcommand{\arraystretch}{1.06}
    \resizebox{\linewidth}{!}{%
    \begin{tabular}{@{}L{0.18\linewidth}L{0.08\linewidth}L{0.08\linewidth}L{0.10\linewidth}L{0.16\linewidth}L{0.06\linewidth}L{0.20\linewidth}@{}}
    \toprule
    \textbf{Setting} & \textbf{Encoder} & \textbf{E layers} & \textbf{Tree} & \textbf{Synapses / branch} & \textbf{Seeds} & \textbf{Notes} \\
    \midrule
    Gradient fidelity / path-gain / noise-resilience mechanism & id.+ReLU & $[128,128]$ & $[3,3]$ & $N_E{=}40$, $N_I \in \{0,5,10,20,40\}$ & $5$ fidelity, $5$ oracle & Main mechanism sweeps in Figs.~\ref{fig:gradient_fidelity} and \ref{fig:mechanistic}; learned bounded tanh reactivation. \\
    MNIST / figure-ground local performance & id.+ReLU & $[128]$ & $[3,3]$ & $N_E{=}40$, $N_I{=}20$ & $5$ & Best local 5F matched-width/scalar-fallback runs in Table~\ref{tab:gap_closing}; learned bounded tanh reactivation. \\
    Fashion-MNIST performance & id.+ReLU & $[128]$ & $[3,3]$ & $N_E{=}40$, $N_I{=}20$ & $5$ & Dedicated five-seed standard-vs-local performance sweep; learned bounded tanh reactivation. \\
    Cue-routing PV-LocalCA & router & $[64]$ & $[2]$ & $N_E{=}10$, $N_I{=}4$ & $5$ & Learned router with two pathway groups and rank-$2$ structured feedback; learned bounded tanh reactivation. \\
    CIFAR-10 compact direct-I-stream control ladder & id.+ReLU & $[20]$ E, no I cells & $[3,3,3,3]$ & $N_E{=}25$, $N_I{=}25$ direct I-to-E, with matched no-I controls & $5$ per condition & Harder-dataset stress test with input-driven inhibitory conductance, decoder $[32,16]$, decoder-aware soma mapping, and additive fairness controls. \\
    Morphology$\times$inhibition appendix map & id.+ReLU & $[128,128]$ & $[2,2]$, $[3,3]$, $[4,4]$, $[2,2,2]$, $[3,3,3]$ & $N_E{=}40$, $N_I \in \{0,5,10,20,40\}$ & $3$ & Supportive 3-seed sweep; learned bounded tanh reactivation. \\
    \bottomrule
    \end{tabular}%
    }
    \renewcommand{\arraystretch}{1.0}
    \caption{\textbf{Representative architectures used in the manuscript.} All dendritic rows use explicit branch-level excitatory and inhibitory inputs without direct somatic synapses, nonnegative conductances, and a nonnegative first-layer drive enforced by an explicit ReLU transfer stage in the main runs. The tree column gives dendritic morphology per excitatory neuron, and the synapse counts report excitatory and inhibitory synapses per branch.}
    \end{table}
    
    \subsection{Hyperparameters}
    Table~\ref{tab:hyperparameters} collects the core training hyperparameters for each main experiment. Learning rates are applied through parameter groups (top-$k$, dendritic block-linear, reactivation, decoder) with the values given in the representative config files; where a single LR is reported, every group uses that value.
    
    \begin{table}[htbp]\centering\scriptsize
    \renewcommand{\arraystretch}{1.06}
    \setlength{\tabcolsep}{3pt}
    \begin{tabular}{@{}L{0.26\linewidth}L{0.07\linewidth}L{0.07\linewidth}L{0.10\linewidth}L{0.10\linewidth}L{0.07\linewidth}L{0.22\linewidth}@{}}
    \toprule
    \textbf{Setting} & \textbf{Opt.} & \textbf{LR} & \textbf{Batch} & \textbf{Epochs} & \textbf{Weight decay} & \textbf{Notes} \\
    \midrule
    MNIST / Fashion-MNIST / figure-ground MNIST (5F MW/scalar-fallback LocalCA) & Adam & $10^{-3}$ / $5\cdot10^{-4}$ (block / react.) & 256 & 100 & 0 & fixed-epoch, no early stopping; learned bounded tanh reactivation \\
    MNIST / Fashion-MNIST / figure-ground MNIST (matched backprop reference) & Adam & $10^{-3}$ / $5\cdot10^{-4}$ & 256 & 100 & 0 & matched to LocalCA setup \\
    Gradient fidelity / path-gain / noise resilience & Adam & $10^{-3}$ & 256 & 50 & 0 & 5 seeds, hooks enabled for diagnostic capture \\
    Morphology $\times$ inhibition regime map (supportive) & Adam & $10^{-3}$ & 256 & 50 & 0 & 3 seeds \\
    Cue routing & Adam & $10^{-3}$ & 256 & 90 & 0 & early stopping with patience 30 \\
    CIFAR-10 compact direct-I-stream depth-4 LocalCA & Adam & $10^{-3}$ / $10^{-4}$ (block / react.) & 256 & 400 & 0 / 0.01 & grad-clip 5.0, early stop patience 50, nonlinear decoder \\
    CIFAR-10 compact direct-I-stream depth-4 BP & Adam & $10^{-3}$ & 256 & 200 & 0.01 & early stop patience 40, matched backpropagation reference \\
    \bottomrule
    \end{tabular}
    \renewcommand{\arraystretch}{1.0}
    \caption{\textbf{Training hyperparameters by experiment.} All runs use Adam with default $\beta_1{=}0.9$, $\beta_2{=}0.999$, $\varepsilon{=}10^{-8}$. Reactivation and block-linear parameter groups use a slower LR to keep conductance-level dynamics stable. Representative configs are provided in the accompanying repository.}
    \label{tab:hyperparameters}
    \end{table}
    
    \subsection{Effective Broadcast Dimensionality in Main Architectures}
    Table~\ref{tab:effective_broadcast_dimensionality} reports the effective feedback dimensionality of each main architecture.
    
    \begin{table}[htbp]\centering\small
    \renewcommand{\arraystretch}{1.06}
    \begin{tabular}{@{}L{0.29\linewidth}L{0.17\linewidth}L{0.12\linewidth}L{0.32\linewidth}@{}}
    \toprule
    \textbf{Setting} & \textbf{Soma/core dim.} & \textbf{Output classes} & \textbf{Main feedback field} \\
    \midrule
    Gradient fidelity / path-gain / noise resilience & $128,128$ & $10$ & Vector at width-matched stage; scalar at wider branch stages \\
    MNIST / Fashion-MNIST performance & $128$ & $10$ & Vector at width-matched stage; scalar at wider branch stages \\
    Figure-ground MNIST performance & $128$ & $10$ & Vector at width-matched stage; scalar at wider branch stages \\
    Cue routing & $64$ & $2$ & Tuned matched-width/scalar-fallback baseline; higher-rank tests use low-rank or pathway-vector feedback \\
    Compact CIFAR-10 harder-dataset stress test & $20$ & $10$ & Decoder-aware vector at matched width; scalar at wider branch stages \\
    \bottomrule
    \end{tabular}
    \renewcommand{\arraystretch}{1.0}
    \caption{\textbf{Effective dimensionality of the main feedback condition.} Output loss gradients are first mapped into decoder-input/soma coordinates. The matched-width/scalar-fallback mode preserves the vector only when the current dendritic-stage width equals the soma/core dimension and otherwise falls back to a single scalar per example. It is therefore not ancestry-shared soma feedback. Global scalar, ancestry-shared, low-rank, pathway-vector, and transported-error conditions are labeled separately as controls or oracles.}
    \label{tab:effective_broadcast_dimensionality}
    \end{table}
    
    \subsection{Compute Resources}
    All experiments were run on an institutional GPU cluster. Reproduction runs use one NVIDIA A100 or H100-class GPU (40--80GB memory) per seed and do not require distributed training. Table~\ref{tab:compute_resources} gives approximate single-GPU runtime ranges for the reported experiment families; unreported exploratory runs used additional cluster time and are not counted in the reproduction estimate.

    In a matched one-seed current-code profile on an NVIDIA RTX PRO 6000, the same shunting model, batch size, 180-epoch schedule, hooks, and final evaluation took 190\,s with LocalCA and 125\,s with backpropagation. Process-wide peak CUDA allocation was 1.70\,GiB in both cases (2.03 vs.\ 1.97\,GiB reserved). Thus the present research implementation provides no measured speed or memory advantage; it retains diagnostic state and is not optimized for either. We did not measure hardware energy and make no energy-efficiency claim.
    
    \begin{table}[htbp]
    \centering
    \small
    \begin{tabular}{@{}L{0.39\linewidth}L{0.18\linewidth}L{0.17\linewidth}L{0.16\linewidth}@{}}
    \toprule
    \textbf{Experiment family} & \textbf{Runs} & \textbf{Time / run} & \textbf{Approx. GPU-h} \\
    \midrule
    Performance and verification & $\sim$25 & 5--20 min & 4--9 \\
    Gradient, path-gain, inhibition, and oracle diagnostics & $\sim$80 & 10--30 min & 15--40 \\
    Morphology, stress, rule, and feedback controls & $\sim$120 & 5--35 min & 20--70 \\
    CIFAR-10 control ladder & 70 & 0.8--2.5 h & 60--175 \\
    Figure generation and CPU-side summaries & -- & $<5$ h total & $<5$ \\
    \bottomrule
    \end{tabular}
    \caption{\textbf{Approximate compute for reported results.} Ranges reflect single-seed wall-clock variation across A100/H100-class GPUs, data-loading overhead, and early stopping.}
    \label{tab:compute_resources}
    \end{table}
    
    \subsection{Code, Data, and Asset Availability}
    \label{app:code_data_assets}
    MNIST, Fashion-MNIST, and CIFAR-10 are public datasets or public benchmark assets cited in the paper; we use them under their standard published access conditions and cite their original sources in the bibliography. The accompanying source package contains the model code, training scripts, diagnostic scripts, figure-generation scripts, representative configuration files, tests, precomputed summary files, and analysis entry points used for the manuscript. Code and reproduction scripts will be released publicly after publication with an immutable repository tag. The paper does not release a new dataset or a stand-alone pretrained model asset.
    Dataset access pages are the original MNIST site / UCI entry (\url{https://archive.ics.uci.edu/dataset/683/mnist+database+of+handwritten+digits}), the Fashion-MNIST repository (\url{https://github.com/zalandoresearch/fashion-mnist}), and the CIFAR page (\url{https://www.cs.toronto.edu/~kriz/cifar.html}).
    
    \begin{table}[htbp]
    \centering
    \small
    \setlength{\tabcolsep}{4pt}
    \begin{tabular}{@{}L{0.20\linewidth}L{0.34\linewidth}L{0.36\linewidth}@{}}
    \toprule
    \textbf{Asset} & \textbf{Use in this paper} & \textbf{Access / license or terms note} \\
    \midrule
    MNIST & Digit classification and derived nonnegative synthetic tasks & Public benchmark from the original MNIST site / UCI entry; UCI lists DOI 10.24432/C53K8Q and asks users to follow the original acknowledgement policy; no raw-data redistribution. \\
    Fashion-MNIST & Apparel classification stress test & Public Zalando Research benchmark; repository is MIT licensed; no raw-data redistribution beyond standard dataset loaders. \\
    CIFAR-10 & Flattened harder-dataset stress test & Public CIFAR dataset site asks users to cite Krizhevsky's technical report; no raw-data redistribution. \\
    Synthetic tasks & Figure-ground MNIST, noise resilience, cue integration & Generated procedurally from public benchmarks or random seeds described in Appendix~\ref{app:synthetic_tasks}; no new third-party asset is introduced. \\
    \bottomrule
    \end{tabular}
    \caption{\textbf{Existing assets used in the manuscript.} We credit the original sources in the bibliography, use public benchmark assets under their published access conditions, and do not redistribute third-party raw data in the manuscript package.}
    \end{table}
    
    \subsection{LocalCA Broadcast and Gradient Assignment}
    
    LocalCA first maps the output loss derivative to the soma/core activation-space error $\delta_0^a$ and then constructs a branch voltage-error field. In the final dendritic layer this is the decoder-input error; in earlier matched-width layers the practical implementation reuses the same coordinate by neuron index as an approximate layer-soma teaching signal. This is an inter-layer feedback approximation, separate from the within-tree path-transport approximation studied by the main theory. In the main matched-width/scalar-fallback mode, if the current stage width matches the soma/core error dimension, the vector is reused directly:
    \[
    e_{n,b}^{a}=\delta_{0,b}^{a}\in\mathbb{R}^{d_0}.
    \]
    The stricter scalar control compresses the soma/core error to one value per example,
    \[
    \bar{\delta}_b=\frac{1}{d_0}\sum_{c=1}^{d_0}\delta_{0,bc}^{a},
    \qquad
    e_{n,b}^{a}=\bar{\delta}_b\mathbf{1}_n,
    \]
    where $d_0$ is the dimension of the soma/core error. If dimensions do not match, the implementation falls back to scalar expansion. The ancestry-shared mode instead repeats coordinate $\delta_{0,bu}^{a}$ over the contiguous block of compartments descended from soma $u$ whenever the branch-stage width is an integer multiple of $d_0$. The low-rank controls use fixed random projection and mixing matrices $P_K$ and $Q_n$,
    \[
    c_b=P_K\delta_{0,b}^{a}\in\mathbb{R}^{K},
    \qquad
    e_{n,b}^{a}=Q_n c_b.
    \]
    The pathway-vector controls use router-inferred pathway roles to gate and transport a low-rank role vector. The transported oracle uses the effective tree recursion. If $p$ is the parent of child $c$, activation-space error is propagated by
    \[
    e_{c}^{a}=e_{p}^{a}\, f_p'(V_p) R_{p,\varepsilon}^{\mathrm{tot}} g_{c\to p}^{\mathrm{den}},
    \]
    and the local pre-reactivation voltage error used in conductance eligibility is
    \[
    e_n^V=e_n^a f_n'(V_n),
    \]
    where $f_n'(V_n)=1$ when reactivation is disabled.
    This operational definition is the implementation counterpart of the effective path gain $\tilde{\alpha}_n$ in Theorem~\ref{thm:tree_backprop}.
    
    For shunting branches, the conductance-space LocalCA gradients assigned before raw-parameter transformation are
    \[
    \widehat{\nabla}_{g_{nj}^{E}} L
    =\left\langle e_n^V\, x_{nj}^{E} R_{n,\varepsilon}^{\mathrm{tot}}(1-V_n)\right\rangle_B,
    \qquad
    \widehat{\nabla}_{g_{nj}^{I}} L
    =\left\langle e_n^V\, x_{nj}^{I} R_{n,\varepsilon}^{\mathrm{tot}}(0-V_n)\right\rangle_B,
    \]
    and for a dendritic coupling from child $c$ to parent $n$,
    \[
    \widehat{\nabla}_{g_{c\to n}^{\mathrm{den}}} L
    =\left\langle e_n^V R_{n,\varepsilon}^{\mathrm{tot}}(a_c-V_n)\right\rangle_B.
    \]
    For additive controls, the matching local derivatives are
    \[
    \widehat{\nabla}_{g_{nj}^{E}} L
    =\langle e_n^V x_{nj}^{E}\rangle_B,
    \qquad
    \widehat{\nabla}_{g_{nj}^{I}} L
    =-\langle e_n^V x_{nj}^{I}\rangle_B,
    \qquad
    \widehat{\nabla}_{g_{c\to n}^{\mathrm{den}}} L
    =\langle e_n^V a_c\rangle_B.
    \]
    The optimizer applies the usual descent step; equivalently, the sign can be absorbed into the definition of the broadcast error. After these conductance-space gradients are formed, they are multiplied by the positive-transform derivative $\psi'(\theta)$ and by any active synapse masks before assignment to raw parameters.
    
    \subsection{4F and 5F Local Modulators}
    
    The 3F rule uses only the local eligibility and broadcast error above. The 4F and 5F variants multiply the conductance-space gradient by branch-level reliability factors before raw-parameter transformation. For 4F, the implementation records $V_n$ from each branch layer and the corresponding soma/core activity $V_0$. For each batch item $b$, it first averages over the non-batch coordinates of the recorded layer tensor to obtain scalar summaries $\bar V_{n,b}$ and $\bar V_{0,b}$. It then estimates covariance over the batch axis and smooths the resulting scalar with an EMA:
    \[
    r_n^{\mathrm{4F}}=
    \frac{\mathrm{Cov}_B(\bar{V}_{n,b},\bar{V}_{0,b})}
    {\sqrt{\mathrm{Var}_B(\bar{V}_{n,b})\mathrm{Var}_B(\bar{V}_{0,b})}+\varepsilon}.
    \]
    The EMA is initialized from the first observed batch statistic; the recorded tensors are detached from autograd, and the same mini-batch supplies both the statistics and the gradient update. This proxy is clamped to the positive stability range $[0.1,2.0]$ before it multiplies the local gradient; sign-opposed covariance therefore reduces the multiplier rather than reversing update direction. Because the smoothed numerator and denominator can be updated separately, and because single-sample fallbacks use online covariance moments, the clamped value can exceed one and should be read as a preconditioner rather than a literal Pearson correlation. Single-sample online fallbacks for these statistics use Welford's algorithm \cite{welford1962note}.
    The 5F factor adds a bounded branch-level preconditioner based on how much branch voltage variance remains predictable from parent or soma-level activity:
    \[
    \phi_n=
    \mathrm{clip}_{[0.25,4.0]}
    \left(
    \frac{\mathrm{Var}(V_n)}
    {\sigma_{\mathrm{res},n}^{2}+\varepsilon}
    \right),
    \]
    where $\sigma_{\mathrm{res},n}^{2}$ is estimated online from a scalar ridge regression of the branch-layer voltage on the configured parent proxy. Batch-mode estimates center over examples and flatten remaining layer coordinates before updating the smoothed variance and covariance moments; single-sample fallback uses scalar online moments. Since these moments are smoothed and clamped rather than recomputed as a single ordinary least-squares fit on a fixed batch, $\phi_n$ can transiently fall below one or rise above one. The lower clamp keeps noisy residual estimates from suppressing a branch-layer update entirely, and the upper clamp prevents highly predictable branches from dominating optimization. In the main 5F runs, the multiplier is $r_n^{\mathrm{4F}}\phi_n$. The 5F sensitivity diagnostic in Fig.~\ref{fig:five_factor_sensitivity} varies the clamp and EMA rate to verify that the reported MNIST performance is not a single brittle clamp setting. By default the reliability statistics use EMA rate $\alpha_{\mathrm{EMA}}=0.1$ and a conditional-EMA residual estimator with ridge $\lambda_{\mathrm{ridge}}=10^{-3}$ for $\sigma_{\mathrm{res},n}^{2}$, and all gate denominators use the stabilizer $\varepsilon=10^{-8}$.
    
    \subsection{Algorithm}
    Algorithm~\ref{alg:main_localca_update} gives the update order for the main 5F matched-width/scalar-fallback condition and shows where ancestry-shared, higher-bandwidth, and transported-error controls enter.
    \begin{algorithm}[H]
    \caption{Main 5F Matched-Width/Scalar-Fallback LocalCA Update}\small
    \label{alg:main_localca_update}
    \begin{algorithmic}[1]
    \STATE \textbf{Input:} Dendritic model, batch $(x,y)$, optimizer
    \STATE Forward pass; loss $L$, output error $\delta^y$
    \STATE Somatic/core activation error $\delta_0^a =$ decoder-input Jacobian product $J_{\mathrm{dec}}(h_{\mathrm{core}})^\top \delta^y$ (or $W_{\mathrm{dec}}^\top\delta^y$ for a linear decoder)
    \FOR{each layer $n$ (reverse)}
        \STATE If widths match, set $e_n^a=\delta_0^a$; otherwise broadcast $\bar{\delta}$ to the stage. Ancestry-shared/low-rank/path controls replace this line
        \STATE Update EMA estimates for $r_n^{\mathrm{4F}}$ and $\phi_n$ from batch voltages
        \STATE Set $e_n^V=e_n^a f_n'(V_n)$, with $f_n'(V_n)=1$ when reactivation is disabled
        \STATE $\widehat{\nabla}_{g_j^{\mathrm{syn}}}L \leftarrow r_n^{\mathrm{4F}}\phi_n\langle x_j R_{n,\varepsilon}^{\mathrm{tot}}(E_j-V_n)e_n^V\rangle_B$
        \STATE $\widehat{\nabla}_{g_{c\to n}^{\mathrm{den}}}L \leftarrow r_n^{\mathrm{4F}}\phi_n\langle R_{n,\varepsilon}^{\mathrm{tot}}(a_c-V_n)e_n^V\rangle_B$
        \STATE Map conductance gradients to raw parameters by multiplying by $\partial g/\partial\theta$
    \ENDFOR
    \STATE Optional controls replace line 5 with scalar, rank-$K$, path-structured, or transported-oracle feedback.
    \STATE Clip gradients; optimizer step
    \end{algorithmic}
    \end{algorithm}
    
    \section{Theoretical Details}
    \label{app:theory}
    This section collects theory-adjacent material that supports implementation and interpretation but is not part of the main proof. The central derivation remains Theorem~\ref{thm:tree_backprop}, Corollary~\ref{cor:factorization}, and Prop.~\ref{prop:inhibitory_path_gain}.

    \subsection{Conditional Path-Gain Compression}
    Let $z_n=\log \alpha_n^{\mathrm{cond}}$ be the log conductance-stage path gain for compartment $n$. If added inhibitory conductance contributes a path-dependent attenuation
    \begin{equation}
    \beta_n(x)=\sum_{k\in\mathcal{A}(n)}
    \log\!\left(1+\frac{\Delta G_k^I(x)}{g_k^{\mathrm{tot}}(x)}\right),
    \qquad
    z'_n=z_n-\beta_n,
    \label{eq:log_gain_attenuation}
    \end{equation}
    then, over compartments or examples,
    \begin{equation}
    \Var(z')=\Var(z)+\Var(\beta)-2\Cov(z,\beta).
    \label{eq:log_gain_variance}
    \end{equation}
    Thus shunting narrows the log path-gain field exactly when
    \begin{equation}
    \Cov(z,\beta)>\frac{1}{2}\Var(\beta).
    \label{eq:compression_condition}
    \end{equation}
    Eq.~\eqref{eq:log_gain_attenuation} follows by taking the logarithm of the product ratio in Prop.~\ref{prop:inhibitory_path_gain}; Eq.~\eqref{eq:log_gain_variance} then applies the variance identity for $z-\beta$. This is a diagnostic condition under one averaging measure, not a theorem that inhibition minimizes rank. In finite trained networks, its covariance margin can disagree with coefficient-of-variation and exact-error-rank summaries because those statistics average over different compartment and sample axes.

    \begin{proposition}[Path-gain dispersion controls 3F alignment]
    \label{prop:path_gain_alignment}
    For a single example before batch averaging, consider one soma/tree with the exact voltage-space soma error $\delta_0^V$ supplied, and write the local eligibility for conductance parameter $i$ as $z_i$. Let $n(i)$ be the compartment containing parameter $i$. Use the effective gain $\tilde{\alpha}_{n(i)}$ below; when reactivation is disabled, it equals $\alpha_{n(i)}^{\mathrm{cond}}$. Assume both gradient vectors are nonzero. The exact and 3F ancestry-shared gradients are
    \begin{equation}
    g_i^{\mathrm{BP}}=z_i\tilde{\alpha}_{n(i)}\delta_0^V,
    \qquad
    \widehat{g}_i^{\mathrm{3F}}=z_i\delta_0^V .
    \label{eq:path_gain_alignment_gradients}
    \end{equation}
    If transfer derivatives and conductance couplings are nonnegative, $\tilde{\alpha}_{n(i)}\geq0$, so corresponding nonzero components have the same sign. Their cosine is
    \begin{equation}
    \cos(\widehat{g}^{\mathrm{3F}},g^{\mathrm{BP}})
    =
    \frac{\sum_i w_i\tilde{\alpha}_{n(i)}}
    {\sqrt{\sum_i w_i}\sqrt{\sum_i w_i\tilde{\alpha}_{n(i)}^2}}
    =
    \frac{1}{\sqrt{1+\mathrm{CV}_w(\tilde{\alpha}_{n(i)})^2}},
    \qquad
    w_i=(z_i\delta_0^V)^2 .
    \label{eq:path_gain_alignment_cosine}
    \end{equation}
    \end{proposition}
    \begin{proof}
    Substitute Eq.~\eqref{eq:path_gain_alignment_gradients} into the cosine formula and cancel the common signed eligibility factors through $w_i$. The final equality is the weighted identity $E_w[\tilde{\alpha}^2]=E_w[\tilde{\alpha}]^2(1+\mathrm{CV}_w(\tilde{\alpha})^2)$.
    \end{proof}

    Prop.~\ref{prop:path_gain_alignment} is conditional on the exact layer-soma error. The implemented updates are batch averages of products, so the proposition describes the pre-batch geometry of per-example gradient contributions rather than a guarantee for a practical multi-layer broadcast.
    
    \subsection{Random-Broadcast Alignment Intuition}
    
    Feedback-alignment arguments suggest that random or low-rank feedback can be useful when the induced local update remains positively correlated with the exact gradient \cite{lillicrap2016random}. In this model, that correlation depends on how local eligibility factors co-vary with the conductance-only path gain in Eq.~\eqref{eq:path_sum} and with the reactivation derivatives in Eq.~\eqref{eq:effective_path_gain}. We treat this as intuition only. The central theoretical object in the paper is still the exact pre-reactivation compartment error $\partial L/\partial V_n=\tilde{\alpha}_n\delta_0^V$; when reactivation is disabled, $\tilde{\alpha}_n=\alpha_n^{\mathrm{cond}}$. The empirical question is how faithfully different broadcast fields approximate that quantity.
    
    \subsection{Morphology-Aware Extensions}
    The following variants are implemented controls or architectural extensions used to test whether coarse morphology-aware feedback can approximate exact path transport. They are not required for the main 5F matched-width/scalar-fallback results, but they define the ancestry-shared, path-propagation, depth, normalization, and pathway-vector controls reported in the appendix figures and tables.
    
    \paragraph{Path-integrated propagation.}
    Modulate broadcast error by $\pi_n = \pi_{n-1} \cdot R_{n-1}^{\mathrm{tot}} \cdot \bar{g}_{n-1}^{\mathrm{den}}$, approximating depth attenuation from Eq.~\eqref{eq:path_sum} without computing the exact sample-specific path gain.
    
    \paragraph{Depth modulation.}
    Per-branch scaling $\kappa_j = \kappa_{\mathrm{base}} / (d_j + d_0)$, mirroring cable attenuation and testing whether a simple depth prior can stabilize distal updates.
    
    \paragraph{Dendritic normalization.}
    $\Delta g_j^{\mathrm{den}} \leftarrow \Delta g_j^{\mathrm{den}} / (\sum_k g_k^{\mathrm{den}} + \varepsilon)$, analogous to homeostatic scaling \cite{turrigiano2008homeostatic} and used as an additive-control normalization comparison.
    
    \paragraph{Pathway-vector feedback.}
    For tasks with latent pathway structure, the broadcast becomes a role vector inferred from the router. Local pathway activity gates this vector. Upward block transport then passes it to earlier branches, aligning their feedback with downstream descendants.
    
    \paragraph{Router-derived branch roles.}
    Each branch receives a role profile inferred from router assignments and incoming block weights rather than from a hand-coded apical/basal label.
    Plasticity is then modulated by branch selectivity and a synapse-specific role-alignment factor, emphasizing pathway-consistent updates without imposing a heuristic branch taxonomy.
    
    \subsection{HSIC Auxiliary Objectives}
    
    Following \cite{gretton2005hsic}, the figure-ground MNIST LocalCA runs optionally apply an HSIC-style auxiliary gradient to selected layer representations $Z$:
    \[
    \mathcal{L}^{\mathrm{self}} = B^{-2}\tr(\mathbf{K}_Z\mathbf{H}\mathbf{K}_Z\mathbf{H}),
    \qquad
    \mathcal{L}^{\mathrm{target}} = -B^{-2}\tr(\mathbf{K}_Z\mathbf{H}\mathbf{K}_Y\mathbf{H}).
    \]
    Here $B$ is batch size, $\mathbf{H}$ is the centering matrix, $\mathbf{K}_Z$ uses the configured representation kernel, and $\mathbf{K}_Y$ uses one-hot class labels when \path{target_source=labels}. The reported figure-ground runs use an RBF kernel with fixed bandwidth $\sigma=1.0$, HSIC weight $0.01$, self and target weights $0.3$, five-epoch warmup, and gradient clipping at $0.1$; the code applies the resulting auxiliary gradient through the configured local-learning representation pathway rather than through the conductance-factorization theorem. Moderate weights ($0.01$--$0.1$) improve figure-ground MNIST but have negligible effect on MNIST.
    
    \section{Task Construction and Robustness Protocols}
    \label{app:synthetic_tasks}
    This section specifies the benchmark, synthetic-task, depth, and feedback-noise protocols used in the experiments, with emphasis on which inputs preserve the nonnegative-drive regime assumed by the conductance model.
    
    \subsection{General Data Handling}
    
    MNIST, Fashion-MNIST, and CIFAR-10 are loaded from standard public dataset loaders, flattened before entering the dendritic core, scaled to $[0,1]$, and passed through the configured nonnegative transfer stage in the main experiments. MNIST-derived synthetic tasks preserve the flattened $28\times28$ geometry. Fixed random seeds define task generators such as corruption projections and context masks; generator parameters are shared across train, validation, and test splits, while sample-level corruptions are drawn per example. Validation accuracy selects checkpoints, and test accuracy is reported. CIFAR-10 is used as a flattened harder-dataset stress test under the same conductance constraints rather than as a competitive vision benchmark; no data augmentation is used as a performance device.
    
    \subsection{Supervised Benchmark Tasks}
    
    \paragraph{MNIST and Fashion-MNIST.}
    These ten-class tasks use nonnegative pixel inputs. MNIST tests whether 5F LocalCA with matched-width/scalar-fallback feedback approaches matched backpropagation on clean supervised data; Fashion-MNIST adds a same-dimensional distribution shift.
    
    \paragraph{CIFAR-10 harder-dataset stress test.}
    CIFAR-10 inputs are flattened RGB images in $[0,1]$. The compact CIFAR family uses a smaller dendritic core, nonlinear decoder, and decoder-aware LocalCA mapping to test whether the feedback-fidelity ladder persists on harder data: matched-width/scalar-fallback feedback is weak, higher-rank feedback helps, and exact effective path transport approaches the matched shunting backpropagation reference.
    
    \paragraph{Figure-ground MNIST.}
    For each reshaped $28\times28$ MNIST image, the right half carries the digit and the left half is replaced by independent distractors $u_{r,c}\sim\mathrm{Uniform}(0,0.25)$, clipped to $[0,1]$ when needed. The task tests whether local credit assignment exploits spatial signal/distractor structure without leaving the nonnegative-drive regime. The main LocalCA performance rows use the HSIC auxiliary objective with weight $0.01$, which adds roughly $3$~pp in the ablation; the matched shunting BP reference in Table~\ref{tab:gap_closing} is the standard cross-entropy reference.
    
    \paragraph{Noise resilience.}
    Flattened MNIST images are corrupted as $\tilde{x}=\mathrm{clip}(x+\sigma_{\mathrm{task}}Az,0,1)$ with $\sigma_{\mathrm{task}}=1.5$, a fixed task-seeded projection $A\in\mathbb{R}^{784\times50}$, and fresh $z\sim\mathcal{N}(0,I_{50})$ for each corrupted example. The fixed projection creates correlated interference without reducing the task to memorizing one corruption pattern.
    
    \paragraph{Cue integration.}
    Two noisy cue streams ($A$ and $B$) are presented simultaneously for two-class classification, with a one-hot context indicating which cue is reliable on each trial. Cue vectors are clipped to $[0,1]$ after noise. The fixed-pathway variant duplicates context into both cue branches; the learned-routing variant requires the router to discover cue separation from data. This feedback-rank diagnostic tests structured pathway-vector broadcast inside the nonnegative-input regime.
    
    \subsection{Depth Scaling and Noise Robustness}
    These stress diagnostics keep the same broad model family while changing dendritic depth or corrupting the broadcast signal. For depth scaling, dendritic depth varies from 1--4 layers (branch factors $[9]$ to $[3,3,3,3]$): shunting local degrades from 63.5\% to 57.4\%, additive local falls from 54.9\% to 29.7\%, the shunting advantage grows from $+8.5$ to $+27.7$~pp, and matched backpropagation references remain near $90$--$92\%$ (Fig.~\ref{fig:additional_stress_tests}A). A checkpoint check confirms that every nominal coupling stage remains active in all 40 LocalCA depth-sweep runs; the smallest stage-median coupling is $0.716$, so this result is not explained by a collapsed dendritic stage. For noise robustness, Gaussian noise $\mathcal{N}(0,\sigma^2)$ is added to the broadcast error; shunting remains robust to $\sigma\!\leq\!0.1$ (about $62\%$), while additive drops from 46.5\% to chance at $\sigma{=}1.0$, indicating that shunting credit signals carry useful learning information beyond broadcast magnitude alone (Fig.~\ref{fig:additional_stress_tests}B).


\begin{thebibliography}{99}
    
    \bibitem{koch1999biophysics}
    Koch, C. (1999).
    \emph{Biophysics of Computation: Information Processing in Single Neurons}.
    Oxford University Press.
    
    \bibitem{dayan2001theoretical}
    Dayan, P., \& Abbott, L. F. (2001).
    \emph{Theoretical Neuroscience: Computational and Mathematical Modeling of Neural Systems}.
    MIT Press.
    
    \bibitem{poirazi2003pyramidal}
    Poirazi, P., Brannon, T., \& Mel, B. W. (2003).
    Pyramidal neuron as two-layer neural network.
    \emph{Neuron}, 37(6), 989--999.
    https://doi.org/10.1016/S0896-6273(03)00149-1
    
    \bibitem{london2005dendritic}
    London, M., \& H{\"a}usser, M. (2005).
    Dendritic computation.
    \emph{Annual Review of Neuroscience}, 28, 503--532.
    
    \bibitem{spruston2008pyramidal}
    Spruston, N. (2008).
    Pyramidal neurons: dendritic structure and synaptic integration.
    \emph{Nature Reviews Neuroscience}, 9, 206--221. https://doi.org/10.1038/nrn2286
    
    \bibitem{branco2010single}
    Branco, T., \& H{\"a}usser, M. (2010).
    The single dendritic branch as a fundamental functional unit in the nervous system.
    \emph{Current Opinion in Neurobiology}, 20(4), 494--502. https://doi.org/10.1016/j.conb.2010.07.009
    
    \bibitem{larkum2013cellular}
    Larkum, M. (2013).
    A cellular mechanism for cortical associations: an organizing principle for the cerebral cortex.
    \emph{Trends in Neurosciences}, 36(3), 141--151. https://doi.org/10.1016/j.tins.2012.11.006
    
    \bibitem{urbanczik2014dendritic}
    Urbanczik, R., \& Senn, W. (2014).
    Learning by the dendritic prediction of somatic spiking.
    \emph{Neuron}, 81(3), 521--528.
    
    \bibitem{carandini2012normalization}
    Carandini, M., \& Heeger, D. J. (2012).
    Normalization as a canonical neural computation.
    \emph{Nature Reviews Neuroscience}, 13(1), 51--62.
    
    \bibitem{holt1997shunting}
    Holt, G. R., \& Koch, C. (1997).
    Shunting inhibition does not have a divisive effect on firing rates.
    \emph{Neural Computation}, 9(5), 1001--1013.
    
    \bibitem{vogels2011inhibitory}
    Vogels, T. P., Sprekeler, H., Zenke, F., Clopath, C., \& Gerstner, W. (2011).
    Inhibitory plasticity balances excitation and inhibition in sensory pathways and memory networks.
    \emph{Science}, 334(6062), 1569--1573.
    
    \bibitem{fremaux2016threefactor}
    Fr{\'e}maux, N., \& Gerstner, W. (2016).
    Neuromodulated spike-timing-dependent plasticity, and theory of three-factor learning rules.
    \emph{Frontiers in Neural Circuits}, 9, 85. https://doi.org/10.3389/fncir.2015.00085
    
    \bibitem{hennequin2017inhibitory}
    Hennequin, G., Agnes, E. J., \& Vogels, T. P. (2017).
    Inhibitory plasticity: balance, control, and codependence.
    \emph{Annual Review of Neuroscience}, 40, 557--579. https://doi.org/10.1146/annurev-neuro-072116-031005
    
    \bibitem{lillicrap2016random}
    Lillicrap, T. P., Cownden, D., Tweed, D. B., \& Akerman, C. J. (2016).
    Random synaptic feedback weights support error backpropagation for deep learning.
    \emph{Nature Communications}, 7, 13276.
    
    \bibitem{lillicrap2020backprop}
    Lillicrap, T. P., Santoro, A., Marris, L., Akerman, C. J., \& Hinton, G. (2020).
    Backpropagation and the brain.
    \emph{Nature Reviews Neuroscience}, 21(6), 335--346. https://doi.org/10.1038/s41583-020-0277-3
    
    \bibitem{rumelhart1986learning}
    Rumelhart, D. E., Hinton, G. E., \& Williams, R. J. (1986).
    Learning representations by back-propagating errors.
    \emph{Nature}, 323, 533--536.
    
    \bibitem{nokland2016dfa}
    N{\o}kland, A. (2016).
    Direct feedback alignment provides learning in deep neural networks.
    \emph{NeurIPS}, 29, 1037--1045.
    
    \bibitem{guerguiev2017segregated}
    Guerguiev, J., Lillicrap, T. P., \& Richards, B. A. (2017).
    Towards deep learning with segregated dendrites.
    \emph{eLife}, 6, e22901.
    
    \bibitem{sacramento2018dendritic}
    Sacramento, J., Costa, R. P., Bengio, Y., \& Senn, W. (2018).
    Dendritic cortical microcircuits approximate the backpropagation algorithm.
    \emph{NeurIPS}, 31, 8721--8732.
    
    \bibitem{whittington2019theories}
    Whittington, J. C., \& Bogacz, R. (2019).
    Theories of error back-propagation in the brain.
    \emph{Trends in Cognitive Sciences}, 23(3), 235--250.
    
    \bibitem{richards2019dendritic}
    Richards, B. A., \& Lillicrap, T. P. (2019).
    Dendritic solutions to the credit assignment problem.
    \emph{Current Opinion in Neurobiology}, 54, 28--36.
    
    \bibitem{scellier2017equilibrium}
    Scellier, B., \& Bengio, Y. (2017).
    Equilibrium Propagation: Bridging the Gap between Energy-Based Models and Backpropagation.
    \emph{Frontiers in Computational Neuroscience}, 11, 24. https://doi.org/10.3389/fncom.2017.00024
    
    \bibitem{song2024prospective}
    Song, Y., Millidge, B., Salvatori, T., Lukasiewicz, T., Xu, Z., \& Bogacz, R. (2024).
    Inferring neural activity before plasticity as a foundation for learning beyond backpropagation.
    \emph{Nature Neuroscience}, 27, 348--358. https://doi.org/10.1038/s41593-023-01514-1
    
    \bibitem{gretton2005hsic}
    Gretton, A., Bousquet, O., Smola, A. J., \& Sch{\"o}lkopf, B. (2005).
    Measuring statistical dependence with Hilbert-Schmidt norms.
    \emph{Lecture Notes in Computer Science}, 3734, 63--77. \url{https://doi.org/10.1007/11564089_7}
    
    \bibitem{welford1962note}
    Welford, B. P. (1962).
    Note on a method for calculating corrected sums of squares and products.
    \emph{Technometrics}, 4(3), 419--420. https://doi.org/10.1080/00401706.1962.10490022
    
    \bibitem{turrigiano2008homeostatic}
    Turrigiano, G. G. (2008).
    The self-tuning neuron: synaptic scaling of excitatory synapses.
    \emph{Cell}, 135(3), 422--435.
    
    \bibitem{payeur2021burst}
    Payeur, A., Guerguiev, J., Zenke, F., Richards, B. A., \& Naud, R. (2021).
    Burst-dependent synaptic plasticity can coordinate learning in hierarchical circuits.
    \emph{Nature Neuroscience}, 24(7), 1010--1019.
    
    \bibitem{greedy2022burstccn}
    Greedy, W., Zhu, H. W., Pemberton, J., Mellor, J., \& Ponte Costa, R. (2022).
    Single-phase deep learning in cortico-cortical networks.
    \emph{NeurIPS}, 35, 24213--24225.
    
    \bibitem{haider2021latent}
    Haider, P., Ellenberger, B., Kriener, L., Jordan, J., Senn, W., \& Petrovici, M. A. (2021).
    Latent Equilibrium: A unified learning theory for arbitrarily fast computation with arbitrarily slow neurons.
    \emph{NeurIPS}, 34, 17839--17851.
    
    \bibitem{hinton2022forward}
    Hinton, G. (2022).
    The Forward-Forward Algorithm: Some Preliminary Investigations.
    \emph{arXiv:2212.13345}.
    
    \bibitem{dellaferrera2022pepita}
    Dellaferrera, G., \& Kreiman, G. (2022).
    Error-driven input modulation: Solving the credit assignment problem without a backward pass.
    \emph{Proceedings of the 39th International Conference on Machine Learning}, \emph{Proceedings of Machine Learning Research}, 162, 4937--4955.
    
    \bibitem{lee2015dtp}
    Lee, D.-H., Zhang, S., Fischer, A., \& Bengio, Y. (2015).
    Difference target propagation.
    \emph{Machine Learning and Knowledge Discovery in Databases: ECML PKDD 2015}, 498--515. \url{https://doi.org/10.1007/978-3-319-23528-8_31}
    
    \bibitem{meulemans2021dfc}
    Meulemans, A., Tristany Farinha, M., Garc{\'i}a Ord{\'o}{\~n}ez, J., Vilimelis Aceituno, P., Sacramento, J., \& Grewe, B. F. (2021).
    Credit assignment in neural networks through deep feedback control.
    \emph{NeurIPS}, 34, 4674--4687.
    
    \bibitem{millidge2021predictive}
    Millidge, B., Seth, A. K., \& Buckley, C. L. (2021).
    Predictive coding: A theoretical and experimental review.
    \emph{arXiv:2107.12979}.
    
    \bibitem{koch1983nonlinear}
    Koch, C., Poggio, T., \& Torre, V. (1983).
    Nonlinear interactions in a dendritic tree: localization, timing, and role in information processing.
    \emph{PNAS}, 80(9), 2799--2802. https://doi.org/10.1073/pnas.80.9.2799
    
    \bibitem{silver2010neuronal}
    Silver, R. A. (2010).
    Neuronal arithmetic.
    \emph{Nature Reviews Neuroscience}, 11(7), 474--489.
    
    \bibitem{max2024backprojections}
    Max, K., Kriener, L., Pineda Garc{\'\i}a, G., Nowotny, T., Jaras, I., Senn, W., \& Petrovici, M. A. (2024).
    Learning efficient backprojections across cortical hierarchies in real time.
    \emph{Nature Machine Intelligence}, 6, 619--630. https://doi.org/10.1038/s42256-024-00845-3
    
    \bibitem{ma2024auglocal}
    Ma, C., Wu, J., Si, C., \& Tan, K. C. (2024).
    Scaling supervised local learning with augmented auxiliary networks.
    \emph{International Conference on Learning Representations (ICLR)}.
    
    \bibitem{lv2025dll}
    Lv, C., Xu, J., Lu, Y., Wang, X., Wang, Z., Xu, Z., Yu, D., Du, X., Zheng, X., \& Huang, X. (2025).
    Dendritic Localized Learning: Toward Biologically Plausible Algorithm.
    \emph{Proceedings of the 42nd International Conference on Machine Learning}, \emph{Proceedings of Machine Learning Research}, 267, 41682--41700.
    
    \bibitem{erdogan2025ebd}
    Erdogan, M., Pehlevan, C., \& Erdogan, A. T. (2025).
    Error Broadcast and Decorrelation as a Potential Artificial and Natural Learning Mechanism.
    \emph{Advances in Neural Information Processing Systems}, 38.
    
    \bibitem{kao2024ccl}
    Kao, C.-H., \& Hariharan, B. (2024).
    Counter-Current Learning: A Biologically Plausible Dual Network Approach for Deep Learning.
    \emph{Advances in Neural Information Processing Systems}, 37.
    https://doi.org/10.52202/079017-2265
    
    \bibitem{xiao2017fashionmnist}
    Xiao, H., Rasul, K., \& Vollgraf, R. (2017).
    Fashion-MNIST: A novel image dataset for benchmarking machine learning algorithms.
    \emph{arXiv:1708.07747}.
    
    \bibitem{lecun1998gradient}
    LeCun, Y., Bottou, L., Bengio, Y., \& Haffner, P. (1998).
    Gradient-based learning applied to document recognition.
    \emph{Proceedings of the IEEE}, 86(11), 2278--2324.
    
    \bibitem{krizhevsky2009learning}
    Krizhevsky, A. (2009).
    Learning multiple layers of features from tiny images.
    Technical report, University of Toronto.
    
    \bibitem{francioni2026vectorized}
    Francioni, V., Tang, V. D., Toloza, E. H. S., Ding, Z., Brown, N. J., \& Harnett, M. T. (2026).
    Vectorized instructive signals in cortical dendrites.
    \emph{Nature}, 652, 1254--1263. https://doi.org/10.1038/s41586-026-10190-7
    
    \bibitem{iyer2022activedendrites}
    Iyer, A., Grewal, K., Velu, A., Souza, L. O., Forest, J., \& Ahmad, S. (2022).
    Avoiding catastrophe: Active dendrites enable multi-task learning in dynamic environments.
    \emph{Frontiers in Neurorobotics}, 16, 846219.
    
    \end{thebibliography}
\end{document}